\documentclass[english,aps,pra,groupedaddress,nofootinbib,onecolumn,notitlepage]{revtex4-1}
\usepackage[T1]{fontenc}
\usepackage[latin9]{inputenc}
\usepackage{color}
\usepackage{refstyle}
\usepackage{mathrsfs}
\usepackage{amsthm}
\usepackage{amsmath}
\usepackage{amssymb}
\usepackage{graphicx}
\usepackage{nameref}

\makeatletter

\AtBeginDocument{\providecommand\thmref[1]{\ref{thm:#1}}}
\AtBeginDocument{\providecommand\secref[1]{\ref{sec:#1}}}
\AtBeginDocument{\providecommand\figref[1]{\ref{fig:#1}}}
\AtBeginDocument{\providecommand\subref[1]{\ref{sub:#1}}}
\AtBeginDocument{\providecommand\corref[1]{\ref{cor:#1}}}
\AtBeginDocument{\providecommand\eqref[1]{\ref{eq:#1}}}
\RS@ifundefined{subref}
  {\def\RSsubtxt{section~}\newref{sub}{name = \RSsubtxt}}
  {}
\RS@ifundefined{thmref}
  {\def\RSthmtxt{theorem~}\newref{thm}{name = \RSthmtxt}}
  {}
\RS@ifundefined{lemref}
  {\def\RSlemtxt{lemma~}\newref{lem}{name = \RSlemtxt}}
  {}

\@ifundefined{textcolor}{}
{%
 \definecolor{BLACK}{gray}{0}
 \definecolor{WHITE}{gray}{1}
 \definecolor{RED}{rgb}{1,0,0}
 \definecolor{GREEN}{rgb}{0,1,0}
 \definecolor{BLUE}{rgb}{0,0,1}
 \definecolor{CYAN}{cmyk}{1,0,0,0}
 \definecolor{MAGENTA}{cmyk}{0,1,0,0}
 \definecolor{YELLOW}{cmyk}{0,0,1,0}
}

\theoremstyle{plain}
\newtheorem{thm}{\protect\theoremname}
  \theoremstyle{definition}
  \newtheorem{defn}[thm]{\protect\definitionname}
  \theoremstyle{plain}
  \newtheorem{lem}[thm]{\protect\lemmaname}
  \theoremstyle{plain}
  \newtheorem{cor}[thm]{\protect\corollaryname}
  \theoremstyle{plain}
  \newtheorem*{thm*}{\protect\theoremname}

\makeatother

\usepackage{babel}
  \providecommand{\corollaryname}{Corollary}
  \providecommand{\definitionname}{Definition}
  \providecommand{\lemmaname}{Lemma}
  \providecommand{\theoremname}{Theorem}
\providecommand{\theoremname}{Theorem}

\begin{document}
\global\long\def\ket#1{\left| #1 \right\rangle }

\global\long\def\bra#1{\left\langle #1 \right|}

\global\long\def\braket#1#2{\left\langle #1 | #2 \right\rangle }

\global\long\def\ketbra#1#2{|#1\rangle\!\langle#2|}

\global\long\def\braopket#1#2#3{\bra{#1}#2\ket{#3}}

\global\long\def\Tr{\text{Tr}}

\global\long\def\Pr{\text{Pr} }

\global\long\def\sn#1{{\rm \text{sn}\left(#1\right)}}

\global\long\def\lsn#1{{\rm \text{lsn}\left(#1\right)}}

\global\long\def\mana#1{{\rm \mathscr{M}\left(#1\right)}}

\global\long\def\relent#1{{\rm \text{r}_{\mathcal{M}}\left(#1\right)}}

\global\long\def\mono#1{{\rm \mathcal{M}\left(#1\right)}}

\global\long\def\regmono#1{{\rm \mathcal{M}^{\infty}\left(#1\right)}}

\global\long\def\regrelent#1{{\rm \text{r}_{\mathcal{M}}^{\infty}\left(#1\right)}}

\global\long\def\wignorm#1{{\rm \|}#1\|_{W}}

\global\long\def\strangestate{\ketbra{\mathbb{S}}{\mathbb{S}}}

\global\long\def\norrellstate{\ketbra{\mathbb{N}}{\mathbb{N}}}

\global\long\def\stabset#1{{\rm \text{STAB}\left(\mathcal{H}_{#1}\right)}}

\global\long\def\abs#1{\left|#1\right|}

\title{The Resource Theory of Stabilizer Computation}

\author{Victor Veitch$^{1,2}$}

\author{Seyed Ali Hamed Mousavian$^{3}$}

\author{Daniel Gottesman$^{3}$}

\author{Joseph Emerson$^{1,2}$}

\affiliation{$^{1}$Institute for Quantum Computing, University of Waterloo, Waterloo,
Ontario, Canada, N2L 3G1}

\affiliation{$^{2}$Department of Applied Mathematics, University of Waterloo,
Waterloo, Ontario, Canada, N2L 3G1}

\affiliation{$^{3}$Perimeter Institute for Theoretical Physics,
Waterloo, Ontario, Canada, N2L 2Y5}

\begin{abstract}
Recent results on the non-universality of fault-tolerant gate sets underline the critical role of resource states, such as magic states, to power scalable, universal quantum computation. Here we develop a resource theory, analogous to the theory of entanglement, for resources for stabilizer codes. We introduce two quantitative measures - monotones - for the amount of non-stabilizer resource. As an application we give absolute bounds on the efficiency of magic state distillation. One of these monotones is the sum of the negative entries of the discrete Wigner representation of a quantum state, thereby resolving a long-standing open question of whether the degree of negativity in a quasi-probability representation is an operationally meaningful indicator of quantum behaviour.
\end{abstract}

\maketitle


\section{Introduction}

It is a major open problem in quantum information to determine the
origins of quantum computational speedup. In particular, it is highly
desirable to characterize exactly what resources are required for
quantum computation. Beyond the obvious theoretical significance, a
resolution to this problem is important because actual physical systems
almost never afford us access to arbitrary quantum operations. For
instance, a physical implementation of a many-qubit system may suffer
from low purity, small coherence times or the inability to create
large amount of entanglement. The problem is to determine how best
to perform quantum computation in the face of the operational restrictions
dictated by physical considerations. 

Broadly speaking, operational restrictions divide
any set of operations into two classes: the subset
of operations that are easy to implement and the remainder that are
not. For example, a common paradigm in quantum communication is two
or more spatially separated parties communicating using classical
communication and local quantum operations, considered ``cheap''
resources, supplemented by ``expensive'' resources that require global manipulation of quantum states, such as entanglement or quantum communication.
This division of quantum operations into cheap and expensive parts
motivates the development of a resource theory\cite{Horodecki2012ResourceTheories}.
In the sense just explained, entanglement theory is the resource theory
of quantum \emph{communication}\cite{BennettConcentratingEntanglementRT1,BennettPurifyingEntanglementRT2,HorodeckiEntReview}.
In this paper we develop a resource theory of quantum \emph{computation}.

The major obstacle to physical realizations of quantum computation
is that real world devices suffer random noise when they execute quantum
algorithms. Fault tolerant quantum computation offers a framework to overcome
this problem. Starting from a given error rate for the physical
computation, logical encodings can be applied to create arbitrarily small effective error rates for the logically encoded computation. Transversal unitary gates, i.e. gates that
do not spread errors within each code block, play a critical role
in fault tolerant quantum computation. Recent theoretical work has
shown that a fault tolerant scheme with a set of quantum gates that
is both universal and transversal does not exist\cite{Eastin2009Restrictions}. 

Many --- though not all --- of the known fault tolerant schemes
are built around the stabilizer formalism. Stabilizer codes pick out
a distinguished set of preparations, measurements, and unitary transformations
that have a fault tolerant implementation; these are sometimes called ``stabilizer operations''. In this
case the fault tolerant operations are not only sub-universal but actually efficiently classically
simulable by the Gottesman-Knill theorem\cite{Gottesman1998GKtheorem}.
Thus to achieve universal quantum computation the stabilizer operations
must be supplemented with some other fault-tolerant non-stabilizer resource.

A celebrated scheme for overcoming this limitation is the magic state
model of quantum computation \cite{Shor1996FT,Gottesman1999Teleportation}
where the additional resource is a set of ancilla systems prepared
in some (generally noisy) non-stabilizer quantum state. The idea
is to consume non-stabilizer resource states using only stabilizer operations in order to
implement non-stabilizer unitary gates, thereby promoting stabilizer
computation to universal quantum computation. Typically the ancilla
preparation process will be subject to the physical error rates as,
by necessity, this process is outside the realm of the stabilizer formalism.
Thus we expect the raw resource states to be highly mixed, but such states
are not directly useful for injection. The resolution is to perform ``magic
state distillation'' \cite{Bravyi_magic_state_paper}, wherein stabilizer operations are used to distill
a large number of these highly mixed resource states into a small number
of very pure resource states. In this context the power of universal quantum computation  reduces to a characterization of the usefulness of the resource states.

We will divide the set of quantum states into those that can be prepared
using the stabilizer formalism, the \emph{stabilizer }states, and
those that can not, the \emph{magic }states%
\footnote{This somewhat whimsical name stems from two sources. First, the use
of the magic moniker in the original Bravyi and Kitaev paper to describe
states that are, apparently magically, both distillable and useful
for state injection. Second, the long held desire by one of the present
authors to refer to himself as a mathemagician.%
}. The goal is to characterize the optimal use of stabilizer operations
to transform resource magic states $\rho_{\text{res}}$ into the target
magic states $\sigma_{\text{target}}$ required for implementing non-stabilizer
gates. This is best considered as two distinct problems:
\begin{enumerate}
\item 
Starting from any number of copies of  a particular  resource state $\rho_{\text{res}}$, 
is it possible to produce even a single copy of a target state $\sigma_{\text{target}}$?
\item Assuming this process is possible, how efficiently can it be done?
That is, how many copies of $\rho_{\text{res}}$ are required to produce
$m$ copies $\sigma_{\text{target}}$? 
\end{enumerate}
The known protocols are able to distill some, but not all, resource
magic states $\rho_{\text{res}}$ to target states useful for quantum
computation. Until very recently it wasn't even known whether some
distillation protocol could be found to take any magic state to a
nearly pure magic state. Astonishingly, the answer to the first question
(at least in odd dimensions) is no: it was shown in \cite{veitch2012BoundMagicNJP}
that there is a large class of \emph{bound magic states} that are
not distillable to pure magic states using any protocol. (There has
also been some interesting progress on this problem in the qubit case \cite{Campbell_bound_magic_states,Reichardt_magic_state_dist_2009,Reichardt_magic_state_dist_2005}.)
The second question is the primary focus of this work. We devise  quantitative measures of how magic a quantum
state is, allowing us to upper bound the distillation
efficiency. For example, suppose the target state is five times as
magical as the resource state according to such a measure.  Then we can
immediately infer that at least five resource states will be required
for each copy of the target state.

Finding distillation protocols to minimize the amount of resources
required is an extremely important problem. Currently stabilizer codes
provide the best hope for practical quantum computation, but the physical
resource requirement for known distillation protocols is enormous.
For example, reference \cite{FowlerPrimerOnSurfaceCodes} analyzes
the requirements for using Shor's algorithm to factor a 2000 bit number
using physical qubits with realistic error rates%
\footnote{Physical qubit error rate 0.1\%, ancilla preparation error rate 0.5\%.}.
 A surface code construction is used to achieve fault tolerance,
from which it is found that roughly a billion physical qubits are
required. About $94\%$ of these physical qubits are used in the distillation
of the ancilla states required to perform the non-stabilizer gates.
More efficient distillation protocols are critical for the realization
of quantum computation, and there has been a recent flurry of effort
on this front eg. \cite{EastinToffoliMSD,BravyiMagicStateLowOverhead,FowlerPrimerOnSurfaceCodes,JonesFourierDistProtocols,MeierMSD4Qubit}.
Of particular interest is ref. \cite{CampbellMSDAllPrimeDim} showing
how magic state distillation can be extended from qubits to systems
of arbitrary prime dimension (qudits) and giving evidence that distillation
efficiencies may be significantly improved using odd-prime dimensional qudits. Unfortunately,
although these innovations offer improvement over the original magic
state distillation protocols, the physical requirements remain extravagant.
Moreover, it is unclear whether these protocols are near optimal or
if dramatic improvements might still be made. The current work partially addresses
this problem by developing a theory for the characterization of
resources for stabilizer computation.

To quantify the amount of magic resource in a quantum state we introduce
the notion of a \emph{magic monotone}. This is any function mapping
quantum states to real numbers that is non-increasing under stabilizer
operations. This is just the common sense requirement that the amount
of non-stabilizer resource available cannot be increased using only
stabilizer operations. Magic monotones are valid measures of the magic
of a quantum state in exactly the same way entanglement monotones
are valid measures of the entanglement of a quantum state. The main
contribution of this paper is the identification and study of two magic
monotones: the \emph{relative entropy of magic} and the \emph{mana}.

The relative entropy of magic is the analogue of the relative entropy
of entanglement. Both magic theory and entanglement theory belong
to a broader set of resource theories\cite{Horodecki2012ResourceTheories}.
In the general setting the quantum states are divided into the free
states that can be created using the restricted operations and the
resource states that can not. The study of resource theories has
primarily focused on the question of the reversible asymptotic inconvertibility
of resource states, i.e., transformations among many copies of these
states in the limit that an infinite number of copies are available.
General resource theories quantify the usefulness of a quantum
state via monotones that are non-increasing under the restricted class
of operations. 

In general there can be many valid choices of monotone.
In the context of reversible asymptotic interconversion, one standard choice is
the asymptotic regularization of the relative entropy distance to
the set of free states. The relative entropy distance between two
states is $S\left(\rho\|\sigma\right)=\Tr\left(\rho\log\rho\right)-\Tr\left(\rho\log\sigma\right)$.
In the present context the relative entropy monotone is the \emph{relative
entropy of magic} $\relent{\rho}\equiv\min_{\sigma\in\text{STAB}(\mathcal{H}_{d})}S\left(\rho\|\sigma\right)$,
the minimum relative entropy distance between the resource state and
any stabilizer state. We show that in general this monotone is subadditive
in the sense $\relent{\rho^{\otimes2}}<2\relent{\rho}$; because of
this, in the asymptotic regime this measure should be regularized as
$\regrelent{\rho}=\lim_{n\rightarrow\infty}\relent{\rho^{\otimes n}}/n$.
This monotone is the \emph{regularized relative entropy of magic}.
\Secref{Relative-Entropy} is devoted to proving that this monotone
has the property that if it is possible
to reversibly asymptotically interconvert states $\rho$ and $\sigma$
using stabilizer protocols then the rate at which this can be done
is given by $\regrelent{\rho}/\regrelent{\sigma}$. Along the way
we also use the relative entropy of magic to find some interesting
features of magic theory. In particular, we establish that if we wish to
create many copies of any magic state (including a bound magic state)
starting with pure magic states, the ratio of the number of starting
pure states to the number of final magic states is non-zero, even asymptotically.

The generality of the relative entropy distance is both a strength
and a weakness. It offers powerful insight into the similarities between
magic theory and other resource theories. However, by the same token
it can tell us little about the unique features of magic theory. Moreover,
the practical relevance of this monotone is specific to the context of reversible interconversion
of magic states in the asymptotic regime of infinite resources. In the context of magic
state computation, we are most interested in the one-way distillation
of magic states using finite resources. This leads us to expect that
the relative entropy of magic and its regularization may offer limited
practical insight for the problem of magic state distillation. There
is an even more discouraging problem with this measure: like the
relative entropy of entanglement, it appears to be prohibitively computationally difficult
to compute the relative entropy of magic. Moreover, we do not even have a guaranteed algorithm to find the value
of the regularized relative entropy of magic. 
Thus this monotone is useful for the holistic study
of the resource theory of magic but is of little use for giving concrete
bounds on achievable rates of distillation. 

In \secref{Computable_Measure} we introduce a computable measure
of the magic of a quantum state: the \emph{mana}. This monotone is inspired by the usefulness of the discrete Wigner function\cite{Gross2006Hudsons,GrossPrimeDimWignerFunction,Wootters1987_discrete_wig_function}
in previous work showing the existence of bound magic states\cite{veitch2012BoundMagicNJP}.  We will restrict attention to qudits of odd prime dimension, as in the previous work.
There it was shown that negative Wigner representation is a necessary
condition for a magic state to be useful for distillation protocols.
It is natural to wonder if this purely binary negative vs. positive
condition could be extended to a quantitative measure of magic. We
show that this is possible by proving that the sum of the negative
entries of the Wigner representation of a state is a magic monotone,
$\sn{\rho}$. This monotone is intuitively appealing, but it still
has a non-additive composition law. To recover additivity we define a closely related quantity, the
\emph{mana} $\mana{\rho}=\log\left(2\sn{\rho}+1\right)$, for which it follows that 
$\mana{\rho\otimes\sigma}=\mana{\rho}+\mana{\sigma}$.
Since it is easy to explicitly find the Wigner representation of an
arbitrary quantum state it is also easy to compute the mana and find
explicit bounds on the efficiency of magic state distillation; to
distill $m$ copies of a target state $\sigma$ from $n$ copies of
a resource state $\rho$ at least $n\ge m\frac{\mana{\sigma}}{\mana{\rho}}$
copies are required on average. As an application we compute the
mana efficiencies of the distillation protocols studied in \cite{AnwarQutritMSD,Campbell_bound_magic_states}. Our monotone suggests the possibility of protocols offering dramatic improvements in efficiency.
Additionally, we provide a detailed characterization of the mana for the qutrit state space, which includes identifying two distinct states with maximal mana.

\section{Background and Definitions}

\subsection{Stabilizer Formalism}

The stabilizer formalism is critical for the results of the present
paper. Here we will give a very brief overview of the elements of
the theory we require. For an overview of the stabilizer formalism
in the context of fault tolerance see \cite{Gottesman1997Stabilizer,GottesmanIntroErrorCorr}.
For an overview of the phase space techniques for the stabilizer formalism
see \cite{Gross2006Hudsons,Gross2007Evenly}. Reference \cite{veitch2012BoundMagicNJP}
gives an overview of the particular mathematical elements that will
be important for the present paper.

We begin by defining the generalized Pauli operators for prime dimension
and we will build up the formalism from these. Let $d$ be a prime
number and define the boost and shift operators:
\begin{eqnarray*}
X\ket j & = & \ket{j+1 \bmod d}\\
Z\ket j & = & \mbox{\ensuremath{\omega}}^{j}\ket j,\ \omega=\exp\left(\frac{2\pi\imath}{d}\right).
\end{eqnarray*}
From these we can define the generalized Pauli (Heisenberg-Weyl) operators
in prime dimension:
\[
T_{(a_{1},a_{2})}=\begin{cases}
\imath^{a_{1}a_{2}}Z^{a_{1}}X^{a_{2}} & (a_{1},a_{2})\in\mathbb{Z}_{2}\times\mathbb{Z}_{2}, \ d = 2 \\
\omega^{-\frac{a_{1}a_{2}}{2}}Z^{a_{1}}X^{a_{2}} & (a_{1},a_{2})\in\mathbb{Z}_{d}\times\mathbb{Z}_{d},\ d\ne2
\end{cases},
\]
where $\mathbb{Z}_{d}$ are the integers modulo $d$. Note that a
slightly different definition is required for qubits. For a system
with composite Hilbert space $H_{a}\otimes H_{b}\otimes\dots\otimes H_{u}$,
the Heisenberg-Weyl operators may be written as: 
\begin{align*}
T_{(a_{1},a_{2})\oplus(b_{1},b_{2})\dots\oplus(u_{1},u_{2})} & \equiv T_{(a_{1},a_{2})}\otimes T_{(b_{1},b_{2})}\dots\otimes T_{(u_{1},u_{2})}.
\end{align*}

The Clifford operators $\mathcal{C}_{d}$ are the unitaries that,
up to a phase, take the Heisenberg-Weyl operators to themselves, i.e.,
\[
U\in\mathcal{C}_{d}\iff\forall\boldsymbol{u}\exists\phi,\boldsymbol{u}':\ UT_{\boldsymbol{u}}U^{\dagger}=\exp\left(i\phi\right)T_{\boldsymbol{u}'}.
\]
The set of such operators form a group --- this is the Clifford group
for dimension $d$. The pure stabilizer states for dimension $d$ are
defined as
\[
\left\{ S_{i}\right\} =\left\{ U\ket 0:U\in\mathcal{C}_{d}\right\} ,
\]
and we take the full set of stabilizer states to be the convex hull
of this set:
\[
\text{STAB}\left(\mathcal{H}_{d}\right)=\left\{ \sigma\in L\left(\mathcal{H}_{d}\right):\ \sigma=\sum_{i}p_{i}S_{i}\right\} ,
\]
where $p_{i}$ is some probability distribution.

We define stabilizer operations to be any combination of computational
basis  preparation, computational basis measurement, and Clifford rotations.
In particular, this includes all stabilizer state preparations and
measurements. This set of operations defines the "stabilizer subtheory", which is a convex subtheory of the full set of allowed quantum operations on a finite-dimensional system. The only stabilizer measurement we consider directly  is measurement
in the computational basis. The other measurements in the stabilizer subtheory can be generated, in the usual Heisenberg picture,  by conjugation under Clifford rotations.

\subsection{Wigner Functions}

In \secref{Computable_Measure}, we will need the discrete Wigner function\cite{Gross2006Hudsons,Wootters1987_discrete_wig_function},
which is defined for quantum systems with finite, odd Hilbert space dimension. The discrete Wigner function is a direct analog of the usual
infinite-dimensional Wigner function\cite{Wigner1932On}. The idea of such representation 
is to attempt to map quantum theory (states, transformations, and measurements) onto a classical probability theory over
a phase space, which can be any continuous or discrete set. In any such representation some
quantum states and measurements must be mapped to distributions with
negative entries\cite{Ferrie2008Frame,Ferrie2010Necessity}, i.e., negative
``quasi-probabilities'' are unavoidable. 
The discrete Wigner representation for odd dimensions enjoys the special property that all stabilizer
operations can be represented non-negatively, so the Wigner representation gives a
classical probability model for the full stabilizer subtheory.

The discrete Wigner representation of a state $\rho\in L(\mathbb{C}^{d^{n}})$
is a quasi-probability distribution over $\left(\mathbb{Z}_{d}\times\mathbb{Z}_{d}\right)^{n}$,
which can be thought of as a $d^{n}$ by $d^{n}$ grid (see \figref{qutrit_wigs} in \secref{Computable_Measure}).
The mapping assigning quantum states $\rho$ to Wigner functions $\left\{ W_{\rho}\left(\boldsymbol{u}\right)\right\} $
is given by
\[
W_{\rho}(\boldsymbol{u})=\frac{1}{d^{n}}\Tr(A_{\boldsymbol{u}}\rho),
\]
where $\left\{ A_{\boldsymbol{u}}\right\} $ are the \emph{phase space
point operators}. These are defined in terms of the Heisenberg-Weyl
operators as, 
\begin{eqnarray*}
A_{\boldsymbol{0}} & = & \frac{1}{d^{n}}\sum_{\boldsymbol{u}}T_{\boldsymbol{u}},\ A_{\boldsymbol{u}}=T_{\boldsymbol{u}}A_{\boldsymbol{0}}T_{\boldsymbol{u}}^{\dagger}.
\end{eqnarray*}
 These operators are Hermitian so the discrete Wigner representation
is real-valued. There are $\left(d^{n}\right)^{2}$ such operators
for $d^{n}$-dimensional Hilbert space, corresponding to the $\left(d^{n}\right)^{2}$
points of discrete phase space. 

A quantum measurement with POVM $\{E_{k}\}$ is represented by assigning
conditional (quasi-) probability functions over the phase space to
each measurement outcome, 
\[
W_{E_{k}}(\boldsymbol{u})=\Tr(A_{\boldsymbol{u}}E_{k}).
\]
 In the case where $W_{E_{k}}(\boldsymbol{u})\ge0\;\forall\boldsymbol{u}$,
this can be interpreted classically as an indicator function or ``fuzzy measurement'' associated with the probability of getting
outcome $k$ given that the ``physical state'' of the system is at phase space point $\boldsymbol{u}$,
$W_{E_{k}}(\boldsymbol{u})=\text{Pr}(\text{outcome }k|\text{location }\boldsymbol{u})$. 

We say a state $\rho$ has positive representation if $W_{\rho}(\boldsymbol{u})\ge0\ \forall\boldsymbol{u}\in\mathbb{Z}_{d}^{n}\times\mathbb{Z}_{d}^{n}$
and negative representation otherwise. We will say a measurement with
POVM $M=\{E_{k}\}$ has positive representation if $W_{E_{k}}(\boldsymbol{u})\ge0\ \forall\boldsymbol{u}\in\mathbb{Z}_{p}^{n}\times\mathbb{Z}_{p}^{n},\ \forall E_{k}\in M$
and negative representation otherwise. We are now ready to state a
few salient facts about the discrete Wigner representation\cite{Gross2006Hudsons,Gibbons2004Discrete}:
\begin{enumerate}
\item (Discrete Hudson's theorem) A pure state $\ket S$ has positive representation
if and only if it is a stabilizer state. Since convex combinations
of positively represented states also have positive representation
this means, in particular, for any stabilizer state $S$ it holds
that $\Tr(A_{\boldsymbol{u}}S)\ge0\ \forall\boldsymbol{u}$.
\item Clifford unitaries act as permutations of phase space. This means
that if $U$ is a Clifford then 
\[
W_{U\rho U^{\dagger}}(\boldsymbol{v})=W_{\rho}(\boldsymbol{v}'),
\]
for each point $\boldsymbol{v}$. Only a small subset of the possible
permutations of phase space correspond to Clifford operations (namely,
the symplectic ones\cite{Gross2006Hudsons}).
\item The trace inner product is given as $\Tr(\rho\sigma)=d^{n}\sum_{\boldsymbol{u}}W_{\rho}\left(\boldsymbol{u}\right)W_{\sigma}\left(\boldsymbol{u}\right)$. 
\item The phase space point operators in dimension $d^{n}$ are tensor products
of $n$ copies of the $d$ dimension phase space point operators,
eg. $A_{\left(0,0\right)\oplus\left(0,0\right)}=A_{\left(0,0\right)}\otimes A_{\left(0,0\right)}$. 
\item The phase point operators satisfy $\Tr\left(A_{\boldsymbol{u}}\right)=1$.
This implies $\Tr\left(B\right)=\sum_{\boldsymbol{u}}W_{B}\left(\boldsymbol{u}\right)$
for a Hermitian operator $B$.
\item $\rho = \sum_{\boldsymbol{u}} W_\rho (\boldsymbol{u}) A_{\boldsymbol{u}}$.
\end{enumerate}
This is all we need to know about the discrete Wigner function for
the present work. For a much more detailed overview see \cite{GrossPrimeDimWignerFunction,Gross2006Hudsons}
or for a moderately more detailed overview see \cite{veitch2012BoundMagicNJP}.

\subsection{Magic Monotones\label{sec:Magic-Monotones}}

In this paper we are concerned with the transformation of non-stabilizer
states using stabilizer operations. In the same way that a state is
defined to be entangled if it is not separable we define:
\begin{defn}
A state is magic if it is not a stabilizer state.
\end{defn}
The most general kind of stabilizer operation possible is of the following
type:
\begin{defn}
A stabilizer protocol is any map from $\rho\in\mathcal{S}(\mathcal{H}_{d^{n}})$
to $\sigma\in\mathcal{S}(\mathcal{H}_{d^{m}})$ composed from the
following operations:
\begin{enumerate}
\item Clifford unitaries, $\rho\rightarrow U\rho U^{\dagger}$
\item Composition with stabilizer states, $\rho\rightarrow\rho\otimes S$
where $S$ is a stabilizer state
\item Computational basis measurement on the final qudit,
$\rho\rightarrow\left(\mathbb{I}\otimes\ketbra ii\right)\rho\left(\mathbb{I}\otimes\ketbra ii\right)/\Tr\left(\rho\mathbb{I}\otimes\ketbra ii\right)$
with probability $\Tr\left(\rho\mathbb{I}\otimes\ketbra ii\right)$
\item Partial trace of the final qudit, $\rho\rightarrow\Tr_{n}\left(\rho\right)$
\item The above quantum operations conditioned on the outcomes of measurements or classical randomness
\end{enumerate}
\end{defn}

Stabilizer protocols encompass magic state distillation protocols
as an important special case. For a function to be a valid measure
of magic, i.e., a monotone, it must be non-increasing under stabilizer operations, a
requirement that can be formalized as:
\begin{defn}
For each $d$, let $\mathcal{M}_{d}:\ \mathcal{S}\left(\mathcal{H}_{d}\right)\rightarrow\mathbb{R}$
be a mapping from the set of density operators on $\mathcal{H}_{d}\cong\mathbb{C}^{d}$
to the real numbers. Define $\mono{\rho}\equiv \mathcal{M}_{d}\left(\rho\right)\ \forall\rho\in\mathcal{S}\left(\mathcal{H}_{d}\right)$ (for the appropriate $d$)
so that $\mono{\cdot}$ is defined for all finite-dimensional Hilbert
spaces. If, on average, $\mono{\Lambda\left(\rho\right)}\le\mono{\rho}$
for any stabilizer protocol $\Lambda$ then we say $\mono{\cdot}$
is a magic monotone.
\end{defn}
There are two important points to notice here. The first is that one need only
require operations to not increase magic on average; if
$\Lambda\left(\rho\right)=\sigma_{i}\ \text{with probability}\ p_{i}$
then we only require $\mono{\rho}\ge\sum_{i}p_{i}\mono{\sigma_{i}}$.
In particular this means that post selected measurement can increase
magic in the sense that we allow $\mono{\left(\mathbb{I}\otimes\ketbra ii\right)\rho\left(\mathbb{I}\otimes\ketbra ii\right)/\Tr\left(\rho\mathbb{I}\otimes\ketbra ii\right)}\ge\mono{\rho}$
as long as measurement outcome $i$ is obtained with sufficiently
small probability. This allows us to analyze non-deterministic protocols.
The second point is that we do not require convexity, i.e., it can
happen that $\mono{p\rho+\left(1-p\right)\sigma}\ge p\mono{\rho}+\left(1-p\right)\mono{\sigma}$.
Although convexity is a desirable feature it is possible to have interesting and useful 
monotones that are not convex (for example, the logarithmic entanglement
negativity\cite{PlenioLogarithmicNegativity}).  

Convexity constrains the behavior of the monotone on all mixtures of density matrices.  The definition of a magic monotone only requires that the measure be non-increasing on mixtures which are formed via stabilizer operations, and only non-increasing relative to the starting states.  For instance, we might form a mixture $\rho = p \rho_0 + (1-p) \rho_1$ by beginning with the state $\rho_0 \otimes \rho_1$ and discarding the second state with probability $p$ and the first state with probability $1-p$.  A magic monotone must have the property that
\[
\mathcal{M} (\rho_0 \otimes \rho_1) \geq \mathcal{M} (\rho),
\]
whereas convexity requires that 
\[
p \mathcal{M}(\rho_0) + (1-p) \mathcal{M} (\rho_1) \geq \mathcal{M} (\rho).
\]
Even if $\mathcal{M}$ is additive (i.e., $\mathcal{M} (\rho_0 \otimes \rho_1) = \mathcal{M} (\rho_0) + \mathcal{M} (\rho_1)$), the latter equation is a stronger constraint.

Notice also that because Clifford gates and composition with stabilizer
states are reversible within the stabilizer formalism (by the inverse
gate and the partial trace respectively) any monotone must actually
be invariant under these operations, as opposed to merely non-increasing.

\section{Relative Entropy of Magic\label{sec:Relative-Entropy}}

Generic resource theories can, and usually do, admit more than one
valid choice of monotone. Requiring a function to be non-increasing
under stabilizer operations is the minimal imposition for it to be
a sensible measure of magic. However, there is no guarantee that all
monotones will be equally interesting or useful. This leads us to
wonder if some further natural conditions could be imposed to eliminate
some of these measures and pick out especially interesting monotones.
Resource theories are concerned with the problem of using restricted
operations to convert between different types of resource states,
for example distilling pure magic states from mixed ones or changing
one type of pure magic state to another type of pure magic state.
Most often this conversion is studied in the asymptotic regime (eg.
\cite{Gour_rel_entropy_of_frameness,Horodecki2012ResourceTheories,Horodecki_Entang_Thermo,HorodeckiPurityTheory,GourSpekkensAsymmetryRT})
where an infinite number of resource states are assumed to be available
to conversion protocols and the task is to determine the rate at which
one type of resource can be converted into another. In this regime
it turns out that for many resource theories the monotone zoo can
be cut down in a rather spectacular fashion: there is a monotone that
\emph{uniquely} specifies the rate at which the asymptotic interconversion
of resource states can take place. Because of the importance of asymptotic
interconversion of resource states this measure is often called the
unique measure of the resource\cite{Horodecki2012ResourceTheories}.
For magic theory the analogous quantity is the regularized relative
entropy of magic.  The purpose of this section is to introduce this
quantity; however, for magic theory it is not unique.

The relative entropy distance between quantum states is $S\left(\rho\|\sigma\right)\equiv\Tr\left(\rho\log\rho\right)-\Tr\left(\rho\log\sigma\right)$.
This is a measure of how distinguishable $\rho$ is from $\sigma$.
Qualitatively, we might expect a measure of how distinguishable
$\rho$ is from all stabilizer states to be a good measure of magic. This 
 inspires the definition:
\begin{defn}
Let $\rho\in\mathcal{S}\left(\mathcal{H}_{d}\right)$.  Then the relative
entropy of magic is $\relent{\rho}\equiv\min_{\sigma\in\stabset d}S\left(\rho\|\sigma\right)$. 
\end{defn}
The intuition that this should be a magic measure is immediately validated:
\begin{thm}
The relative entropy of magic is a magic monotone.\end{thm}
\begin{proof}
This is essentially a consequence of the nice properties of the relative
entropy and holds for the same reasons that the relative entropy is
a monotone for other resources theories. See \subref{rel_ent_monotone}
for details.
\end{proof}
 The main importance of the relative entropy of magic is in the asymptotic
regime.  This is because the relative entropy of magic is subadditive
in the sense $\relent{\rho^{\otimes n}}<n\relent{\rho}$. This follows from the fact that  in general there can be some entangled stabilizer state $\sigma_{AB}\in\mathcal{S}\left(\mathcal{H}_{d}\otimes\mathcal{H}_{d}\right)$
such that $S\left(\rho\otimes\rho\|\sigma_{AB}\right)<\min_{\sigma_{A},\sigma_{B}\in\text{STAB}\left(\mathcal{H}_{d}\right)}S\left(\rho\otimes\rho\|\sigma_{A}\otimes\sigma_{B}\right)$.
In particular this means that the amount of magic added from adding an extra copy of $\rho$ depends 
on how many copies of $\rho$ we
already have. In the asymptotic limit an appropriate measure
should give the amount of magic in $\rho$ when an infinite number
of copies of $\rho$ are available. This prompts us to introduce the
asymptotic variant of the relative entropy measure:
\begin{defn}
Let $\rho\in\mathcal{S}\left(\mathcal{H}_{d}\right)$.  Then the regularized
relative entropy of magic is $\regrelent{\rho}\equiv\lim_{n\rightarrow\infty}\frac{1}{n}\relent{\rho^{\otimes n}}$. 
\end{defn}
We do not have an analytic expression for the relative entropy of
magic and thus we do not have an analytic expression for the asymptotic
version. Moreover, because of the infinite limit in the definition
we do not even know how to numerically approximate $\regrelent{\rho}$
in general. This is the same as the situation in entanglement theory
where it remains a famous open problem to find a ``single letter''
expression for the regularized relative entropy of entanglement. Nevertheless,
the (regularized) relative entropy of magic is useful for the holistic
study of magic theory. For instance, we will use it as a tool to show
that an asymptotically non-zero amount of pure magic resource states is always required
to produce bound magic states via stabilizer protocols, even though no pure
magic can be extracted from the states that are produced.

\subsection{Relative entropy of magic}

One of the major difficulties with the study of resource monotones
is that the actual computation of the value of the monotone for a
particular state is often an intractable problem. Although we do not
know a simple analytic expression for the relative entropy of magic,
it can be computed numerically. For systems with low Hilbert space
dimension this is reasonably straightforward. The relative entropy
is a convex function and we want to perform minimization over the
convex set of stabilizer states. This means that a simple numerical
gradient search will succeed in finding $\min_{\sigma\in\stabset d}S\left(\rho\|\sigma\right)$.
Each qudit stabilizer state can be written as a convex combination
of the $N$ pure qudit stabilizer states. A simple strategy for finding
the relative entropy of magic is to do a numerical search over the
$N-1$ values that define the probability distribution over the stabilizer
states. Unfortunately, for a system of $n$ qu$d$its the number of
pure stabilizer states is\cite{Gross2006Hudsons}
\[
N=d^{n}\prod_{i=1}^{n}\left(d^{i}+1\right),
\]
and this grows too quickly for a numerical search to be viable in
general. For example, the original $H$-type magic state distillation
protocol\cite{Bravyi_magic_state_paper} consumes $15$ resource states
$\rho_{\text{input}}$ to produce a more pure magic state $\rho_{\text{output}}$.
In principal we can bound the quantity of the resource required via,
\[
\relent{\rho_{\text{input}}^{\otimes15}}\ge p_{i}\relent{\rho_{\text{output}}},
\]
where $p_{i}$ is the success probability of the protocol, but this
would require a numerical optimization over more than $2^{136}$ parameters
using the approach just outlined, which is not viable.

For arbitrary resource states it is not clear how to avoid the numerical
optimization. However, the states typically used in magic state distillation
protocols have a great deal of additional structure that can be exploited.
In particular, many protocols have a ``twirling'' step where a random
Clifford unitary is applied to the resource state to ensure it has
the form,
\[
\rho_{\epsilon}=\left(1-\epsilon\right)\ketbra MM+\epsilon\frac{\mathbb{I}}{d}.
\]
If the twirling map is $\mathcal{T}:\rho_{\text{resource}}\rightarrow\sum_{i}p_{i}U_{i}\rho_{\text{resource}}U_{i}^{\dagger}$
for some subset $\left\{ U_{i}\right\} $ of the Clifford operators,
then 
\begin{eqnarray*}
\min_{\sigma\in\text{STAB}}S\left(\left[\left(1-\epsilon\right)\ketbra MM+\epsilon\frac{\mathbb{I}}{d} \right] \Big\|\sigma\right) & \ge & \min_{\sigma\in\text{STAB}}S\left(\mathcal{T}\left(\left(1-\epsilon\right)\ketbra MM+\epsilon\frac{\mathbb{I}}{d}\right) \Big\| \mathcal{T}\left(\sigma\right)\right)\\
 & = & \min_{p\le p_{T}}S\left(\left[\left(1-\epsilon\right)\ketbra MM+\epsilon\frac{\mathbb{I}}{d}\right] \Big\| \left[p\ketbra MM+\left(1-p\right)\frac{\mathbb{I}}{d}\right] \right)\\
 & = & S\left( \left[\left(1-\epsilon\right)\ketbra MM+\epsilon\frac{\mathbb{I}}{d}\right] \Big\| \left[p_{T}\ketbra MM+\left(1-p_{T}\right)\frac{\mathbb{I}}{d}\right] \right),
\end{eqnarray*}
where $p_{T}$ is the largest value such that $p_{T}\ketbra MM+\left(1-p_{T}\right)\frac{\mathbb{I}}{d}$
is a stabilizer state. This means that the relative entropy of magic
can be computed exactly for states of this form by finding $p_{T}$.
Unfortunately the twirling is only applied to individual qudits so
this does not by itself resolve the numerical problems. Nevertheless,
it is possible to give naive bounds according the following observation:
\begin{eqnarray*}
\relent{\rho_{\text{output}}} & \le & \relent{\rho_{\text{input}}^{\otimes n}}\\
 & \le & n\relent{\rho_{\text{input}}}
\end{eqnarray*}
where we have used the obvious fact that the relative entropy of magic
is subadditive.

This bound might not seem naive at all. One might suspect that the
relative entropy of magic is genuinely additive so $\relent{\rho_{\text{input}}^{\otimes n}}\stackrel{?}{=}n\relent{\rho_{\text{input}}}$.
This seems like a very desirable feature for a monotone to have: $n$
copies of a resource state should contain $n$ times as much resource
as a single copy. The relative entropy of magic does not have this
feature --- it can be the case that $\relent{\rho^{\otimes2}}\lneq2\relent{\rho}$.
To establish this we consider the qutrit \emph{Strange} state $\strangestate$
defined as the pure qutrit state invariant under the symplectic component
of the Clifford group (see \subref{The-Qutrit-Case}). Twirling by
the symplectic subgroup $\text{Sp}\left(2,3\right)$ of the Clifford group%
\footnote{This is the Clifford group modulo the Heisenberg-Weyl (Pauli) group.%
} has the effect 
\[
\sum_{F}\frac{1}{\abs{\text{Sp}\left(2,3\right)}}U_{F}\rho U_{F}^{\dagger}=\left(1-\epsilon_{\rho}\right)\strangestate+\epsilon_{\rho}\frac{\mathbb{I}_{3}}{3},
\]
so we can use our twirling argument to find $\relent{\strangestate}$
exactly. A numerical search over the two qutrit stabilizer states
turns up a state $\sigma\in\stabset 9$ such that $S\left(\strangestate^{\otimes2}\|\sigma\right)<2\relent{\strangestate}$. 

Note that the relative entropy of \emph{entanglement} is also subadditive. However,
there is a very important difference between the entanglement and
magic relative entropy measures: for \emph{pure} states the relative entropy of entanglement
is an additive measure. This fact is at the heart of the importance
of the relative entropy distance for the theory of entanglement. As
we have just seen this is not true for the relative entropy of magic.

\subsection{The (regularized) relative entropy of magic is faithful\emph{ }}

The relative entropy $S\left(\rho\|\sigma\right)$ is $0$ if and
only if $\rho=\sigma$. It is easy to see that this implies that $\relent{\rho}$
is \emph{faithful} in the sense that $\relent{\rho}>0$ if and only
if $\rho$ is magic. Since $\relent{\cdot}$ is a magic monotone,
if it is possible to create a magic state $\sigma$ from a (pure) resource
state $\ketbra{\psi}{\psi}$ using a stabilizer protocol it must be
the case that $\relent{\ketbra{\psi}{\psi}}\ge\relent{\sigma}$.  Previous work established the existence of a large class of ``bound
magic states''%
\footnote{Called bound universal states in the original paper.%
} that can not be distilled to pure magic states\cite{veitch2012BoundMagicNJP}. Together
these facts imply that to create any magic state, including bound magic states,  a non-zero amount of magic is required.
This means, in general, that the amount of magic that can be distilled
from a resource state is not equal to the amount of magic required
to create it; this is analogous to the well-known
result in entanglement theory that the entanglement of formation is
not equal to the entanglement of distillation.

Because the relative entropy of magic is subadditive it could be that
$\regrelent{\rho}=\lim_{n\rightarrow\infty}\frac{1}{n}\relent{\rho^{\otimes n}}=0$
for some magic state $\rho$. I.e., it is not automatic that the regularized
relative entropy of magic is faithful. For example, in the resource
theory of asymmetry\cite{Gour_rel_entropy_of_frameness}, the regularized
relative entropy measure is $0$ for all states. Happily, for magic
theory the relative entropy is well behaved in the asymptotic regime:
\begin{lem}
\label{lem:faithful}
The regularized relative entropy of magic is faithful in the sense
that $\regrelent{\rho}=0$ if and only if $\rho$ is a stabilizer
state.\end{lem}
\begin{proof}
The proof of this fact is a straightforward application of a theorem
of Piani\cite{piani2009RelEntFaithful} showing that the regularized
relative entropy measure is faithful for all resource theories where
the set of restricted operations includes tomographically complete
measurements and the partial trace. The idea is to define a variant
of the relative entropy distance that quantifies the distinguishability
of states using only stabilizer measurements. This quantity lower
bounds the usual relative entropy of magic so by showing that its
regularization is faithful we get the claimed result. See appendix \ref{sub:faithful-proof}
for details. 
\end{proof}
We will need this result for the proof of Corollary \ref{cor:reg_rel_ent_unique}
showing that the regularized relative entropy gives the optimal rate
of asymptotic interconversion. It also represents a strengthening
of our earlier claim that a non-zero amount of pure state magic is
required to create any magic state. For finite size protocols this
followed from the faithfulness of the relative entropy of magic, as
just explained. The faithfulness of the regularized relative entropy
implies that the creation of magic states by an asymptotic stabilizer
protocol requires resource magic states to be consumed at a non-zero
rate. The analogous problem in entanglement theory was the main motivation
for proving that the regularized relative entropy of entanglement
is faithful\cite{Brandao2010GenQuantumStein,piani2009RelEntFaithful}.

\subsection{Asymptotic interconversion and the regularized relative entropy}

In the scenario of asymptotic state conversion, it is useful to consider an
additional property that a magic measure may possess beyond those
required to make it a magic monotone.  To understand the additional property, it is simplest to first consider
the case of finite state interconversion. Suppose there is some stabilizer
protocol $\Lambda$ that maps $n$ copies of resource state $\rho$
to $m$ copies of a target magic state $\sigma$. Then it must be
the case that $\mono{\rho^{\otimes n}}\ge\mono{\sigma^{\otimes m}}$
for any magic monotone $\mono{\cdot}$. If there is also some other
stabilizer protocol that maps $\sigma^{\otimes m}$ to $\rho^{\otimes n}$ then it must be
the case that $\mono{\rho^{\otimes n}}=\mono{\sigma^{\otimes m}}$,
which conceptually just means that if $\rho$ and $\sigma$ are equivalent
resources then they have the same magic according to any magic measure.
It is rarely possible to exactly interconvert between resource states
with only a finite number of copies available. However, it is often the case
that we can get close to a reversible interconversion; that is, that conversion
from copies of $\sigma$ to approximate copies of $\rho$ and then back to approximate copies of $\sigma$ loses
only an asymptotically negligible number of copies of the state.  Thus, we wish to study measures
that satisfy the requirement that asymptotically
reversibly interconvertible states have the same resource value. 

Typically if we try to convert $\rho^{\otimes n}$ into $m$ copies
of $\sigma$, the stabilizer protocol ($\Lambda_{n}:\mathcal{H}_{d^{n}}\rightarrow\mathcal{H}_{d^{m}}$)
we use will depend on the number $n$ of input states. When converting
$\rho$ to $\sigma$ it is thus necessary to consider a family of
stabilizer protocols $\Lambda_{n}$ taking $\rho^{\otimes n}$ as
input and producing $m(n)$ approximate copies of $\sigma$ with an
error $\|\Lambda_{n}\left(\rho^{\otimes n}\right)-\sigma^{\otimes m\left(n\right)}\|_{1}=\epsilon_{n}$.
In the case that the approximation becomes arbitrarily good in the
asymptotic limit (i.e., $\lim_{n\rightarrow\infty}\|\Lambda_{n}\left(\rho^{\otimes n}\right)-\sigma^{\otimes m\left(n\right)}\|_{1}\rightarrow0$)
we say $\rho$ is \emph{asymptotically convertible} to $\sigma$ at
a rate $R\left(\rho\rightarrow\sigma\right)=\lim_{n\rightarrow\infty}\frac{m(n)}{n}$.
The additional condition we consider is that if $\rho$
is asymptotically convertible to $\sigma$ then
\begin{equation}
\lim_{n\rightarrow\infty}\frac{1}{n}\left[\mono{\rho^{\otimes n}}-\mono{\sigma^{\otimes m\left(n\right)}}\right]\ge0.\label{eq:asympt_condition}
\end{equation}
That is, if asymptotic conversion is possible then on average we must
put in at least as much magic as we get out, up to some $o(n)$ discrepancy.

If it is possible to interconvert between $\sigma$ and $\rho$ at
rates $R\left(\sigma\rightarrow\rho\right)=R\left(\rho\rightarrow\sigma\right)^{-1}$
then we say the two resources are \emph{asymptotically reversibly
interconvertible}. Any magic monotone satisfying the additional condition (\ref{eq:asympt_condition})
gives the rate of asymptotic interconversion according to the following
theorem:
\begin{thm}
\label{thm:asymp_conv_rate}Let $\mono{\cdot}$ be a magic monotone
satisfying the condition given by \eqref{asympt_condition} and define
its asymptotic variant $\regmono{\rho}=\lim_{n\rightarrow\infty}\mono{\rho^{\otimes n}}/n$.
Then if it is possible to asymptotically reversibly interconvert between
magic states $\rho$ and $\sigma$ and $\regmono{\sigma}$ is non-zero
the rate of conversion is given by $R\left(\rho\rightarrow\sigma\right)=\regmono{\rho}/\regmono{\sigma}$. \end{thm}
\begin{proof}
This is a special case of broader theorem that says this result holds
in any resource theory. The result was first proved in \cite{Horodecki_Entang_Thermo}.
That paper dealt primarily with entanglement and missed the requirement
that the regularization of the monotone needs to be non-zero. This
was pointed out in \cite{Gour_rel_entropy_of_frameness}, and the
theorem we state here is essentially the application of their Theorem
4 to magic theory. The only subtlety is that instead of the condition
in \eqref{asympt_condition} they require the monotone to be \emph{asymptotically
continuous}, which means $\lim_{n\rightarrow\infty}\|\Lambda_{n}\left(\rho^{\otimes n}\right)-\sigma^{\otimes m\left(n\right)}\|_{1}\rightarrow0$
implies
\[
\lim_{n\rightarrow\infty}\frac{\mono{\Lambda_{n}\left(\rho^{\otimes n}\right)}-\mono{\sigma^{\otimes m\left(n\right)}}}{1+n}\rightarrow0.
\]
The first step of their proof is to show that this condition
implies \eqref{asympt_condition} so we prefer to start with the weaker requirement directly.\end{proof}
\begin{cor}
\label{cor:reg_rel_ent_unique} If it is possible
to asymptotically reversibly interconvert between magic states $\rho$
and $\sigma$, the rate at which this can be done is $R\left(\rho\rightarrow\sigma\right)=\regrelent{\sigma}/\regrelent{\rho}$,
where $\regrelent{\sigma}$ is the regularized relative entropy.
\end{cor}
\begin{proof}
In \cite{SynakRadtke2006AsymptoticContinuity} it is shown that the
relative entropy distance to any convex set of quantum states is asymptotically
continuous. Since asymptotic continuity implies \eqref{asympt_condition}
and the stabilizer states are a convex set, the relative entropy of
magic is a magic monotone satisfying condition \eqref{asympt_condition}.
Moreover, we showed in \lemref{faithful} that the regularized relative
entropy is non-zero for all magic states. 
\end{proof}
Notice that the relative entropy is only one example of a monotone
satisfying the conditions of \thmref{asymp_conv_rate}. There could
be other monotones for which this result holds. In fact it is conceivable
that this result holds for every magic monotone. For any magic monotone
with this property, if it possible to asymptotically interconvert between
$\rho$ and $\sigma$, it must be the case that $\regmono{\rho}=C\regrelent{\rho}\implies\regmono{\sigma}=C\regrelent{\sigma}$
so $\regrelent{\sigma}/\regrelent{\rho}=\regmono{\sigma}/\regmono{\rho}$.
I.e., the regularization of such magic measures can differ only up to
a multiplicative factor that can vary between sets of quantum states
where asymptotic interconversion is possible. 

If we have a resource measure $\mono{\cdot}$ that is additive then
it will be equal to its own regularization, $\mono{\cdot}=\regmono{\cdot}$.
If this measure also satisfies \eqref{asympt_condition} then it
will tell us how to compute the asymptotic interconversion rate whenever
asymptotic interconversion is possible. In the particular case that we have
a resource theory where asymptotic interconversion is possible between
any two resource states then it is easy to see that up to a constant
factor there really is a single unique measure of the resource. For instance,
this is true of bipartite pure entangled states, and the entanglement entropy\cite{BennettConcentratingEntanglementRT1}
is an additive measure that satisfies our condition.  Thus the entanglement entropy is the genuinely unique
measure of pure state bipartite entanglement. One of the special features
of the relative entropy of entanglement is that it reduces to the
entanglement entropy on pure states. It is this feature which is ultimately
responsible for the privileged status of the relative entropy of entanglement.
In the case of magic theory the relative entropy of magic does not
reduce to an additive measure on pure states so there is no
apparent reason to prefer the relative entropy of magic over any other
monotone satisfying the conditions of \thmref{asymp_conv_rate}.
This stands in contrast to the claim that the relative entropy distance
to the set of free states is the unique measure of the resource (eg.
\cite{Horodecki2012ResourceTheories}).

\subsection{Discussion}

The privileged status of the relative entropy magic comes from its
role in the asymptotic regime. Since the assumption of infinite state
preparations is unreasonable for an actual physical system one might
expect that the measure would be of limited practical value. This
suspicion is lent additional weight by the fact that, like the regularized
relative entropy distance in other resource theories, it is not known
how to compute $\regrelent{\rho}$ in general. The regularized relative
entropy distance is essentially useless for analyzing the performance
of particular distillation protocols. Nevertheless, the monotone is
a useful tool for the holistic study of the resource theory of magic.
This is the role of the regularized relative entropy distance in the
theory of entanglement, where it is a well studied object. We had
a taste of this already in our demonstration that the amount of pure
state magic required to create a magic state does not equal the amount
of pure state magic that can be distilled from that state. It is an
exciting direction for future work to see what other insights can
be gleaned from the relative entropy of magic and its asymptotic variant. 

It is also important to understand exactly what corollary \corref{reg_rel_ent_unique}
says. The statement is that \emph{if} asymptotic interconversion is
possible then the rate is given by $\regrelent{\rho}/\regrelent{\sigma}$.
This ``if clause'' is a deceptively strong requirement: it is not guaranteed
that asymptotic interconversion will always be possible, or even that
it will ever be possible. In particular, every currently known magic
state distillation protocol has rate $0$ and it is an important open
problem to determine if a positive rate distillation protocol exists.

\section{A Computable Measure of Magic\label{sec:Computable_Measure}}

The results of the previous section deal primarily with reversible
conversion of magic states in the limit where infinite copies are
available, but for magic state distillation we are interested in the
one way distillation of resource magic states to pure target magic
states in the regime where only a finite number of resource states
are available. Because of this, there is no reason to prefer the (regularized)
relative entropy of magic over any other monotone. Nevertheless, the
relative entropy, like any monotone, gives non-trivial bounds on distillation
efficiency. But there is a more fundamental problem: it is generally computationally
hard to compute the relative entropy, and we have no idea how
to compute the regularized relative entropy so we are unable to find
explicit upper bounds to distillation. In this section we address
this issue by introducing a computable measure of magic.

\subsection{Sum negativity and mana}

Previous work establishing the existence of bound magic states\cite{veitch2012BoundMagicNJP}
provides a starting place in the search for a computable monotone.
The fundamental tool in that construction is the discrete Wigner function.
There it was found that a necessary condition for a magic state to
be distillable is that it have negative Wigner representation. However,
that work is purely binary in the sense that it does not distinguish
degrees of negative representation. It is natural to suspect that
a state that is ``nearly'' positively represented is less magic
than a state with a large amount of negativity in its representation.
Here we formalize this intuition by showing that the absolute value
of the sum of the negative entries of the discrete Wigner representation
of a quantum state is a magic monotone. 
\begin{defn}
The sum negativity of a state $\rho$ is the sum of the negative elements
of the Wigner function, $\sn{\rho}\equiv\sum_{\boldsymbol{u}:W_{\rho}(\boldsymbol{u})<0}\abs{W_{\rho}\left(\boldsymbol{u}\right)}\equiv\frac{1}{2}\left(\sum_{\boldsymbol{u}}\abs{W_{\rho}\left(\boldsymbol{u}\right)}-1\right)$.

The right hand side of this expression follows because the normalization
of quantum states ($\Tr\rho=1$) implies $\sum_{\boldsymbol{u}}W_{\rho}\left(\boldsymbol{u}\right)=1$.
The advantage of writing the expression in this form is that $\wignorm{\rho}\equiv\sum_{\boldsymbol{u}}\abs{W_{\rho}\left(\boldsymbol{u}\right)}$
is a multiplicative norm and is thus very nice to work with. By this
we mean that the composition law is given as:
\begin{eqnarray}
\wignorm{\rho\otimes\sigma} & = & \sum_{\boldsymbol{u},\boldsymbol{v}}\abs{W_{\rho\otimes\sigma}\left(\boldsymbol{u}\oplus\boldsymbol{v}\right)}\nonumber \\
 & = & \sum_{\boldsymbol{u},\boldsymbol{v}}\abs{W_{\rho}\left(\boldsymbol{u}\right)W_{\sigma}\left(\boldsymbol{v}\right)}\nonumber \\
 & = & \left(\sum_{\boldsymbol{u}}\abs{W_{\rho}\left(\boldsymbol{u}\right)}\right)\left(\sum_{\boldsymbol{v}}\abs{W_{\rho}\left(\boldsymbol{v}\right)}\right).\label{eq:wignorm_composition}
\end{eqnarray}

Since the sum negativity is a linear function of this quantity we
can establish that the former is a magic monotone by showing this
for the latter:\end{defn}
\begin{thm}
The sum negativity is a magic monotone.\end{thm}
\begin{proof}
It suffices to show $\sum_{\boldsymbol{u}}\abs{W_{\rho}\left(\boldsymbol{u}\right)}$
is a magic monotone by verifying the required properties. The main
components are the use of $\rho=\sum_{\boldsymbol{u}}W_{\rho}(\boldsymbol{u})A_{\boldsymbol{u}}$
and the composition identity \ref{eq:wignorm_composition}, which
is the main motivation for working with this quantity rather than
with the sum negativity directly. See appendix \ref{sub:wig_norm_monotone_proof}
for details. 
\end{proof}
The sum negativity is an intuitively appealing way of using the Wigner
function to define a magic monotone, but it has some irritating features.
The worst of these is the composition law,
\[
\sn{\rho^{\otimes n}}=\frac{1}{2}\left[\left(2\sn{\rho}+1\right)^{n}-1\right],
\]
which has the troubling feature that a linear increase in the number
of resource states implies an exponential increase the amount of resource
according to the measure. Happily there is a simple resolution to
this problem suggested by the composition law, \eqref{wignorm_composition}. 
We define a new monotone by a particular function of the sum negativity:
\begin{defn}
The mana of a quantum state $\rho$ is $\mana{\rho}\equiv\log\left(\sum_{\boldsymbol{u}}\abs{W_{\rho}\left(\boldsymbol{u}\right)}\right)=\log\left(2\sn{\rho}+1\right)$.
\end{defn}

\begin{thm}
The mana is a magic monotone. 
\end{thm}
\begin{proof}
Most of the monotone requirements
follow because log is a monotonic function, but there is a small subtlety
here. Consider a stabilizer protocol that sends $\rho\rightarrow\sigma_{i}$
with probability $p_{i}$ (e.g., post-selected computational basis measurement).
Then we require $\log\left(\wignorm{\rho}\right)\ge\sum_{i}p_{i}\log\left(\wignorm{\sigma_{i}}\right)$.
This need not be true for arbitrary monotonic functions of $\wignorm{\rho}$
but it is easy to see that it follows from the concavity of $\log$
and $\wignorm{\rho}\ge\sum_{i}p_{i}\wignorm{\sigma_{i}}$. 
\end{proof}

From \eqref{wignorm_composition} this monotone is additive in the
sense
\[
\mana{\rho\otimes\sigma}=\mana{\rho}+\mana{\sigma}.
\]
Beyond its intuitive appeal, additivity is a nice feature for a monotone
to have because it makes the bound on distillation efficiency take
an especially nice form. How many copies $n$ of a resource magic
state $\rho$ are required to distill $m$ copies of a resource magic
state $\sigma$? Suppose we have a stabilizer protocol $\Lambda\left(\rho^{\otimes n}\right)\rightarrow\sigma_{i}\text{ with probability }p_{i}$,
then the monotone condition combined with additivity shows:
\[
\sum_{i}p_{i}\mana{\sigma_{i}}\le n\mana{\rho}.
\]
Taking $\sigma_{0}=\sigma$ and $p_{0}=p$, the above discussion lets
us see:
\begin{thm}
Suppose $\Lambda$ is a stabilizer protocol that consumes resource
states $\rho$ to produce $m$ copies of target state $\sigma$, succeeding
probabilistically. Any such protocol requires at least $\mathbb{E}\left[n\right]\ge m\frac{\mana{\sigma}}{\mana{\rho}}$
copies of $\rho$ on average.\label{thm:mojo_distill_bound}\end{thm}
\begin{proof}
Suppose $\Lambda\left(\rho^{\otimes k}\right)=\sigma^{\otimes m}\text{ with probability }p$.
The fact that the mana is an additive magic monotone implies:
\[
k\mana{\rho}\ge p \; m\mana{\sigma}\implies\frac{k}{p}\ge m\frac{\mana{\sigma}}{\mana{\rho}}
\]
Letting $l$ be the number of times we must run the protocol to get
a success we have $n=kl$ and, 
\[
\mathbb{E}\left[l\right]=\frac{1}{p},
\]
from which it follows that $\mathbb{E}\left[n\right]=\frac{k}{p}\ge m\frac{\mana{\sigma}}{\mana{\rho}}$.
\end{proof}
We can only bound the average number of copies required because the
monotone is only non-increasing on average under stabilizer operations --- it might increase conditionally on a specific measurement outcome. 

The most common case for magic state distillation is nested distillation protocols, which a little thought will show are covered by the bound
as a special case. Indeed, this bound covers a broader set
of protocols than it might first appear. One might have expected to
do better by ``recycling'' the output states of the failed protocols.
For instance, if 
\[
\Lambda\left(\rho^{\otimes k}\right)=\begin{cases}
\sigma^{\otimes m} & \text{with probability }p\\
\tau & \text{with probability }1-p
\end{cases},
\]
then one expects to reduce the overhead of the total number of copies
$\rho$ required by introducing a second stabilizer protocol:
\[
\mathcal{E}\left(\tau\otimes\rho^{\otimes k'}\right)=\sigma^{\otimes m}\text{ with probability }q.
\]
However, by just combining the two steps we have a new protocol $\tilde{\Lambda}\left(\rho^{\otimes\left(k+k'\right)}\right)=\sigma^{\otimes m}\text{ with probability }\tilde{p}=p+\left(1-p\right)q$
and our theorem applies.

Computing the mana of a quantum state is straightforward: we find
the Wigner function by taking the trace of $\rho$ with the $d^{2}$
phase space point operators and compute the mana directly. This means
that the mana provides a simple way to numerically upper bound the
efficiency of distillation protocols, fulfilling the major promise
of this section.

\subsection{Uniqueness of sum negativity}

Quantifying the magic of a state by the negativity in its Wigner representation
is an intuitively appealing idea, but it is not clear that the sum
of the negative elements is the best way to do this. For example,
we might have instead looked at the maximally negative element of
the Wigner function, $\text{maxneg}\left(\rho\right)=-\min_{\boldsymbol{u}}W_{\rho}\left(\boldsymbol{u}\right)$.
It is not immediately obvious that the sum negativity is a better
way to quantify the magic of a quantum state than the maximal negativity
just defined. However, it turns out that the maximal negativity is not a magic
monotone, so it is not a useful measure of resources for stabilizer
computation. In fact, we will now show that \emph{any} magic monotone
that is determined solely by the values of the negative entries of
the Wigner function (and in particular not by the positions in phase
space of the negative entries) can be written as a function of only
the sum negativity. 

The reason that the maximally negative entry is not a magic monotone
is that it is not invariant under composition with stabilizer states.
Suppose we have some resource state $\rho$ and we compose it with
the maximally mixed state on a qudit $\mathbb{I}_{d}/d$.  Then $\text{maxneg}\left(\rho\otimes\mathbb{I}_{d}/d\right)=-\min_{\boldsymbol{u},\boldsymbol{v}}W_{\rho}\left(\boldsymbol{u}\right)\cdot W_{\mathbb{I}/d}(\boldsymbol{v})=-\min_{\boldsymbol{u},\boldsymbol{v}}W_{\rho}\left(\boldsymbol{u}\right)\cdot\frac{1}{d^{2}}=\text{maxneg}\left(\rho\right)/d^{2}$.
Therefore, this function can decrease under composition with stabilizer states,
and thus can increase under partial trace: it is a poor measure of
the amount of resource in $\rho$. The natural requirement that magic monotones
must be invariant under composition with arbitrary stabilizer states
is an extremely strong one; it forms the backbone of our proof of
the uniqueness of the sum negativity. 
\begin{thm}
\label{thm:sum_neg_unique}Assume $\mono{\rho}$ is a function on
quantum states that satisfies the following conditions: 1. $\mono{\rho}$
is a magic monotone, 2. $\mono{\rho}$ is determined only by the negative
values of the Wigner function, and 3. $\mono{\rho}$ is invariant under
arbitrary permutations of discrete phase space (that is, even under
permutations that do not correspond to quantum transformations). Then
$\mono{\rho}$ may be written as a function of only $\sn{\rho}$. \end{thm}
\begin{proof}
Consider two quantum states $\rho$ and $\rho'$ that have Wigner
representations with different negative entries but $\sn{\rho}=\sn{\rho'}$.
The idea is to construct stabilizer ancilla states $A$ and $A'$
such that $\rho\otimes A$ and $\rho'\otimes A'$ have the same negative
Wigner function entries. In this case conditions 2 and 3 imply $\mono{\rho\otimes A}=\mono{\rho'\otimes A'}$
and since magic monotones are invariant under composition with stabilizer
states this means $\mono{\rho}=\mono{\rho'}$, ie. $\mono{\rho}$
is entirely determined by the sum negativity. For details see appendix \ref{sub:sum_neg_unique_proof}.
\end{proof}
For our proof of \thmref{sum_neg_unique} to succeed it is critical
that the value of the monotone does not depend on the locations of
the negative entries. All magic monotones must be invariant under
Clifford unitaries, $\mono{U\rho U^{\dagger}}=\mono{\rho}\ \forall U\in C_{n}$,
and these operations correspond to permutations of the phase space.
Thus the monotone condition already implies invariance under a subset
of possible permutations (namely those that preserve the symplectic
inner product). However, we require invariance under arbitrary permutations
and there is no compelling reason to expect magic monotones to have
this feature in general. It is not clear whether this additional assumption
was really necessary; it was introduced because actually working with
only the symplectic transformations is extremely challenging. It remains
an interesting open problem to either prove uniqueness without this
assumption or else give a counterexample in the form of a magic monotone
that is determined by just the negative entries of the Wigner representation
and does depend on their position. Even if the latter resolution is the case,
\thmref{sum_neg_unique} is useful because it at least shows that
sum negativity is the unique ``simple'' monotone, in the sense that
computing its magnitude does not depend on the detailed symplectic structure of
phase space. As simplicity of computation is our primary motivation for
the study of Wigner function monotones, this is a significant advantage.

In \secref{Relative-Entropy} we showed that (the regularization of)
any monotone satisfying a certain natural asymptotic condition uniquely
specifies the rate at which asymptotic interconversion of resource
states is possible. Since the mana is additive, it is clearly equal
to its own regularization. Thus if it satisfied the condition given
by \eqref{asympt_condition} we would be able to compute the conversion
rates explicitly. Typically it is usually a stronger property that
is demanded: asymptotic continuity of the monotone. In appendix \ref{sub:mojo_continuity} we show that the mana is not asymptotically continuous. However,
our counterexample leaves open the possibility that the weaker condition
actually required by the theorem holds. It would be very exciting
to either prove or disprove this.

\subsection{Numerical Analysis of Magic State Distillation Protocols}

To illustrate the use of mana in the evaluation of magic state distillation
protocols we have computed the input and output mana of single steps
of several (qudit) magic distillation protocols from the literature
over a large parameter range. Figures \ref{fig:[[5,1,3]] qutrit}
and \ref{fig:[[8,1,3]] qutrit} present qutrit codes from \cite{AnwarQutritMSD}
and \cite{CampbellMSDAllPrimeDim} respectively. Figure \ref{fig:[[4,1,2]] ququint}
presents a ququint ($d=5$) code from \cite{CampbellMSDAllPrimeDim}.
Notice that none of the protocols come close to meeting the mana bound,
which is illustrated as a red line in all three figures. 

\begin{figure}[h]
\includegraphics[width=7cm]{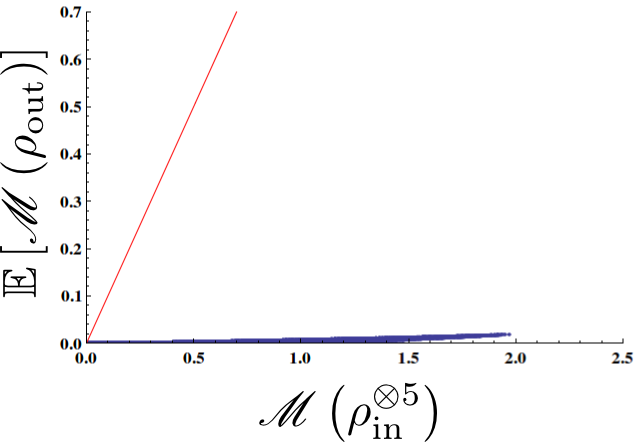}

\caption{Efficiency of the $[[5,1,3]]_3$ qutrit code of \cite{AnwarQutritMSD}.
We generate 50000 inputs of the form $\rho_{\text{in}}=\left(1-p_{1}-p_{2}\right)\ketbra{H_{+}}{H_{+}}+p_{1}\ketbra{H_{-}}{H_{-}}+p_{2}\ketbra{H_{i}}{H_{i}},$
which is the form $\rho_{\text{in}}$ takes after the twirling
step of the protocol. The mana of the $5$ input states is computed
and plotted against the effective mana output following one round
of the protocol, $\mathbb{E}\left[\mathscr{M}\left(\rho_{\text{out}}\right)\right]=\Pr\left(\text{protocol succeeds}\right)\cdot\mana{\rho_{\text{out}}}.$
We used $p_{1}\in_{R}\left[0,0.4\right]$ and $p_{2}\in_{R}\left[0,0.3\right]$,
and the twirling basis states are the eigenstates of the qutrit Hadamard
operator\cite{AnwarQutritMSD}, with eigenvalues $\left\{ 1,-1,\imath\right\} $.
\label{fig:[[5,1,3]] qutrit}}
\end{figure}

\begin{figure}[h]
\includegraphics[width=7cm]{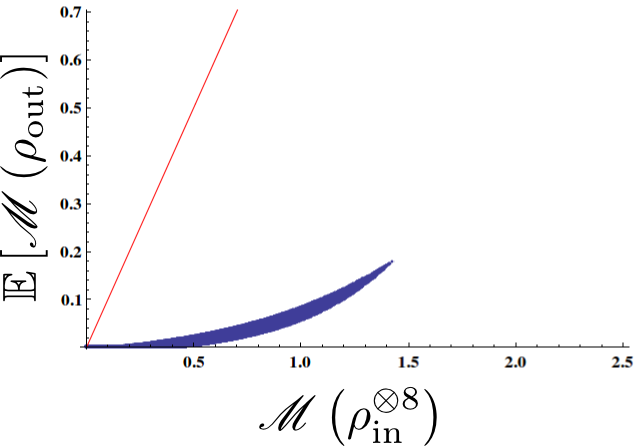}

\caption{Efficiency of the $[[8,1,3]]_3$ qutrit code of \cite{CampbellMSDAllPrimeDim}.  We
generate 50000 inputs of the form $\rho_{\text{in}}=\left(1-p_{1}-p_{2}\right)\ketbra{M_{0}}{M_{0}}+p_{1}\ketbra{M_{1}}{M_{1}}+p_{2}\ketbra{M_{2}}{M_{2}},$
which is the form $\rho_{\text{in}}$ takes after the twirling
step of the protocol. The mana of the $8$ input states is computed
and plotted against the effective mana output following one round
of the protocol, $\mathbb{E}\left[\mathscr{M}\left(\rho_{\text{out}}\right)\right]=\Pr\left(\text{protocol succeeds}\right)\cdot\mana{\rho_{\text{out}}}.$
We used $p_{1}\in_{R}\left[0,0.3\right]$, $p_{2}\in_{R}\left[0,0.3\right]$,
and the twirling basis states are $\ket{M_{0}}=\frac{1}{\sqrt{3}}\left(\mathrm{e}^{\frac{4}{9}\pi\mathrm{i}}\ket 0+\mathrm{e}^{\frac{2}{9}\pi\mathrm{i}}\ket 1+\ket 2\right),\ \ket{M_{1}}=\frac{1}{\sqrt{3}}\left(\mathrm{e}^{\frac{16}{9}\pi\mathrm{i}}\ket 0+\mathrm{e}^{\frac{8}{9}\pi\mathrm{i}}\ket 1+\ket 2\right),\ \ket{M_{2}}=\frac{1}{\sqrt{3}}\left(\mathrm{e}^{\frac{10}{9}\pi\mathrm{i}}\ket 0+\mathrm{e}^{\frac{14}{9}\pi\mathrm{i}}\ket 1+\ket 2\right)$.
\label{fig:[[8,1,3]] qutrit}}
\end{figure}

\begin{figure}[h]
\includegraphics[width=7cm]{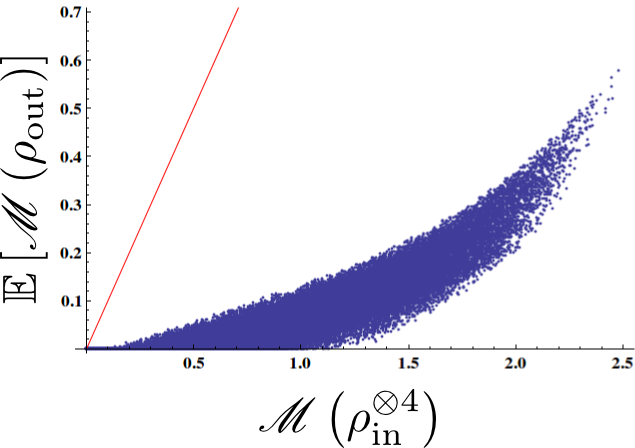}

\caption{Efficiency of the $[[4,1,2]]_5$ ququint code of \cite{CampbellMSDAllPrimeDim}.  We
generate 50000 inputs of the form $\rho_{\text{in}}=\left(1-p_{1}-p_{2}-p_{3}-p_{4}\right)\ketbra{M_{0}}{M_{0}}+\sum_{i=1}^{4}p_{i}\ketbra{M_{i}}{M_{i}},$
which this is the form $\rho_{\text{in}}$ takes after the twirling
step of the protocol. The mana of the $4$ input states is computed
and plotted against the effective mana output following one round
of the protocol, $\mathbb{E}\left[\mathscr{M}\left(\rho_{\text{out}}\right)\right]=\Pr\left(\text{protocol succeeds}\right)\cdot\mana{\rho_{\text{out}}}.$
We used $p_{i}\in_{R}\left[0,0.2\right]$, and the twirling basis states
are the eigenstates of the $CM$ ququint operator defined in \cite{CampbellMSDAllPrimeDim}.
\label{fig:[[4,1,2]] ququint}}
\end{figure}

\subsection{The Qutrit Case\label{sub:The-Qutrit-Case}}

It's interesting to compute the qutrit states with maximal
sum negativity. Since
\begin{eqnarray*}
\sn{\rho} & = & -\sum_{\boldsymbol{u}:\Tr\left(\rho A_{\boldsymbol{u}}\right)<0}\Tr\left(\rho A_{\boldsymbol{u}}\right)\\
 & = & -\Tr\left(\rho\sum_{\boldsymbol{u}:\Tr\left(\rho A_{\boldsymbol{u}}\right)<0}A_{\boldsymbol{u}}\right),
\end{eqnarray*}
it is easy to see that the states with maximal sum negativity must
be eigenstates of operators $\sum_{\boldsymbol{u}\in S}A_{\boldsymbol{u}}$
where $S$ is some subset of the discrete phase space. An exhaustive
search over such subsets reveals two classes of maximally sum negative
states.
\begin{enumerate}
\item The Strange states defined to be those with 1 negative Wigner function entry equal to
$-1/3$. There are ${9 \choose 1}=9$ such states, e.g. 
\[
\ket{\mathbb{S}}=\frac{1}{\sqrt{2}}\begin{pmatrix}0\\
1\\
-1
\end{pmatrix}
\]

\item The Norrell states defined to be those with 2 negative Wigner function entries equal to
$-1/6$. There are ${9 \choose 2}=36$ such states, e.g.
\[
\ket{\mathbb{N}}=\frac{1}{\sqrt{6}}\begin{pmatrix}-1\\
2\\
-1
\end{pmatrix}.
\]

\end{enumerate}
The maximum value is $\sn{\strangestate}=\sn{\norrellstate}=+1/3$.
An example of each type of state is given in \figref{qutrit_wigs}.

\begin{figure}
\includegraphics[width=3cm]{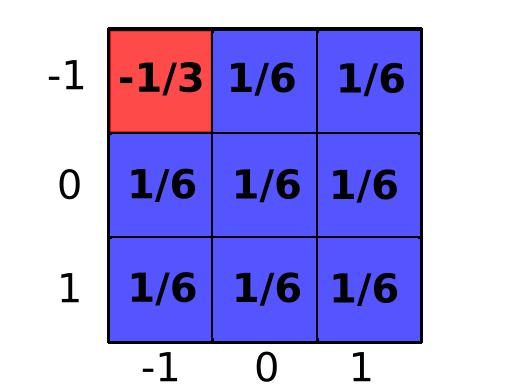}\includegraphics[width=3cm]{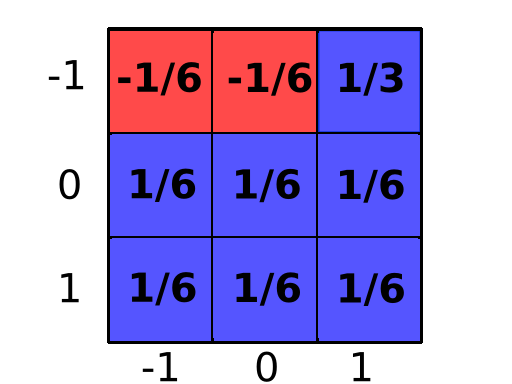}

\caption{The Wigner representations of two qutrit states, $\ket{\mathbb{S}}=\frac{1}{\sqrt{2}}\left(\ket 1-\ket 2\right)$
(left) and $\ket{\mathbb{N}}=\frac{1}{\sqrt{6}}\left(-\ket 0+2\ket 1-\ket 2\right)$
(right). $\ket{\mathbb{S}}$ has sum negativity $\abs{-\frac{1}{3}}$
and the $\ket{\mathbb{N}}$ has sum negativity $\abs{-\frac{1}{6}-\frac{1}{6}}=\frac{1}{3}$.
\textcolor{red}{\label{fig:qutrit_wigs}}}
\end{figure}

Geometrically each Strange state lies outside a single face of the
Wigner simplex and each Norrell state lies outside the intersection
of two faces, analogous to the qubit T-type (outside a face) and H-type
(outside an edge) states. This analogy is further strengthened since
the Norrell states are also the generalized H-type states of \cite{HowardQuditPiOver8}
and \cite{AnwarQutritMSD}.

Note that the states with maximal resource value do not need to agree
between monotones. In particular,
\[
\frac{\relent{\strangestate}}{\relent{\norrellstate}}\approx1.71.
\]
Of course this still leaves open the possibility that $\regrelent{\strangestate}=\regrelent{\norrellstate}$.

\begin{figure}
\includegraphics[scale=0.6]{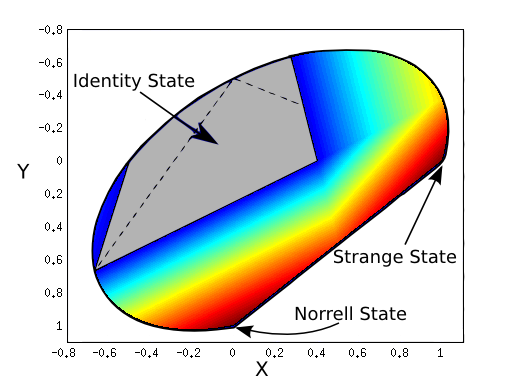}

\caption{The plane $\left(1-x-y\right)\frac{\mathbb{I}}{3}+x\strangestate+y\norrellstate$.
The heat map shows the value of the mana. The light grey ($0$ mana)
region is the set of states in the Wigner simplex, ie. states with positive Wigner representation.
The stabilizer polytope is delineated by a dashed line.}
\end{figure}

\subsection{Wherefore the discrete Wigner function?}

Our main motivation for studying the mana is that it can be computed
explicitly to give concrete bounds on the rate at which magic states
can be converted. However, one might suspect that this bound, although
non-trivial, is rather arbitrary. For example, it is not clear a priori
if the bound given by \thmref{mojo_distill_bound} can ever be saturated,
or under what circumstances this might occur. The mana arose very
naturally from the negativity of the discrete Wigner function, but it is not immediately
clear that the Wigner negativity  is the relevant tool
for the study of magic theory. However, a number of recent results
show that the negativity of the discrete Wigner representation is extremely well
motivated in this context. For example,  could we have started with some other notion
of quasi-probability representation \cite{Ferrie2008Frame} and defined a monotone from that?
Recent work \cite{HowardContextualityandComputation} has shown (at least for small prime dimension) that this
is not the case by connecting the onset of negativity in the DWF with the onset of a non-contextuality violation. 
This means that the subtheory of quantum theory consisting of elements with positive
discrete Wigner representation is a maximal classical subtheory
in the sense of non-classicality given by contextuality. That is,
the set of states with positive discrete Wigner representation is the largest
possible subtheory of quantum theory that includes the stabilizer
measurements and admits a non-contextual hidden variable theory. In
particular this means that any other choice of quasi-probability representation (that represents the stabilizer subtheory non-negatively) would have a positively represented region that is strictly contained
within the discrete Wigner function we use here.

For the purposes of magic state distillation we are more interested in the notion of
non-classicality given by universal quantum computation. The results
of \cite{veitch2012BoundMagicNJP,Mari2013_pos_wig_efficient_sim} show that there
is an intimate connection: the hidden variable model afforded by the
discrete Wigner function leads naturally to an efficient classical
simulation scheme for quantum circuits with positive Wigner representation.
It is not known if access to any negatively represented state suffices
to promote stabilizer computation to universal quantum computation,
but it is at least apparent that the known classical simulation protocols
cannot be extended to deal with this case. In the context of magic
state computation it is desirable for the magic measures to give an
indication of how useful a state is for quantum computation. In this
sense, the fact that the mana is not a faithful monotone is a feature
rather than a bug --- it picks specifically the set of quantum states that do not admit an efficient simulation scheme under stabilizer operations. 

Although the mana is essentially the unique symmetric monotone arising from
the negativity of the Wigner function, it is not the only choice of
monotone arising from the Wigner function. In particular, one very
natural choice is the relative entropy distance to the set of states
with positive Wigner representation, ${\rm \text{r}_{\mathcal{W}}\left(\rho\right)=\min_{\sigma:W_{\sigma}\left(\boldsymbol{u}\right)\ge0\,\forall\boldsymbol{u}}S\left(\rho\|\sigma\right)}$.
It is easy to check that all of the results of \secref{Relative-Entropy}
go through for this new monotone, subject to obvious modifications
in the statement of the theorems.

\subsection{Discussion}

The major inspiration for the monotones of this section was earlier
work showing that states with positive Wigner representation cannot
be distilled by stabilizer protocols. In the theory of entanglement
it is known that states with positive partial transpose (ppt) cannot 
be distilled by LOCC protocols \cite{HorodeckMafiaBoundEnt}, and
this inspired the introduction of the entanglement negativity $\mathcal{N}\left(\rho\right)$,
a measure of the violation of the ppt condition, as a measure of entanglement \cite{VidalWernerComputableEntanglement}.
As with the sum negativity, the major advantage of this measure is
that it is computable, allowing for explicit upper bounds on the efficiency
of entanglement distillation. The entanglement negativity grows exponentially
in the number of resource states, prompting the definition of an additive
variant $\mathcal{LN}\left(\rho\right)\equiv\log\left(2\mathcal{N}\left(\rho\right)+1\right)$
--- exactly as in the present case. Like the mana this measure has the
strange features that it is neither convex nor asymptotically continuous.%
\footnote{In fact it is now known that these two features are closely related
\cite{PlenioLogarithmicNegativity}.%
} The close analogy we have uncovered suggests that it may be possible
to adapt much of the work on entanglement negativity to the magic
case: this is an interesting direction for future work. 

There is at least one way in which the sum negativity is better behaved
than the entanglement negativity. All separable states are local,
but this does not mean that all entangled states are non-local in the
sense that they enable violation of a Bell inequality. In \cite{PeresPPTNonLocConjecture}
Peres conjectured that any ppt state should admit a local hidden variable
model; proving or disproving this conjecture is one of the major outstanding
problems in the study of entanglement. In our case the equivalent
conjecture would be that any state with positive Wigner representation
admits a non-contextual hidden variable model. But in our case the answer is obvious: the
Wigner itself is this non-contextual hidden variable theory! Moreover,
as noted above, recent work \cite{HowardContextualityandComputation} has shown 
(at least for small prime dimension) that magic states admit such a model
\emph{only if} they have positive Wigner representation. The direct resolution of this question (which has proven difficult to solve for other resource theories) is a consequence
of our use of the Wigner function (quasi-probability) technology. However, the
quasi-probability techniques used in this section have no known analogue
in other resource theories. The possibility of exporting this technology to the study
of other resource theories, in particular entanglement theory, is
a fascinating and promising direction for future work. 

A closely related problem is to determine a qubit analogue for the
mana. Because it is possible to violate a contextuality inequality
(eg. a GHZ inequality) using qubit stabilizers, there can be no qubit
analogue for the discrete Wigner function (see also \cite{Wallman2012QubitQuasiProb}).
This is because the discrete Wigner function is a non-contextual hidden
variable theory. Nevertheless, it may be possible to find a computable
monotone of a similar flavour.

\section{Discussion}

In this paper we have introduced the resource theory of magic, showing
how the tools of resource theories can be applied to study the extra
resources required to promote stabilizer computation to universal
quantum computation. In particular, we have introduced the concept
of magic monotone and given two examples: the relative entropy of
magic and the mana. 

The relative entropy of magic and its asymptotic variant are useful
tools for the holistic study of magic theory. In particular, we saw
that (even asymptotically) to create any magic state by consuming 
pure magic states via stabilizer operations a non-zero amount of pure magic states are required.
This established, in conjunction with the results
of \cite{veitch2012BoundMagicNJP}, that generally the amount of magic
that can be extracted from a magic state is not equal to the amount
required to create it: the magic of creation does not equal the magic
of distillation. The main motivation for studying the relative entropy
of magic was that its asymptotic regularization gives the correct rate for
asymptotic interconversion of magic states. However, as we saw, this is not a special
feature of the relative entropy of magic but a (potentially) common
feature among magic monotones. This is promising because the relative entropy of magic
has some serious drawbacks. Foremost among these are the lack of a
closed form expression and the fact that it is a subadditive monotone, even for pure magic
states. The combination of these two irritants implies that computing
the relative entropy of magic generally requires a numerical search
that is computationally infeasible.

To address this shortcoming we introduced the mana, a computable
monotone. We have shown this monotone has the appealing feature that
it is additive, $\mana{\rho\otimes\sigma}=\mana{\rho}+\mana{\sigma}$.
As a consequence, we may give explicit lower bounds on the number of
resource states $\rho$ required to produce $m$ copies of a resource
state $\sigma$. This is an explicit, absolute upper bound on the
efficiency of magic state distillation protocols. This monotone is
in some sense the unique measure of magic arising from the negativity
of the discrete Wigner function. Since the discrete Wigner function
itself is essentially the unique maximal classical representation for the
stabilizer formalism \cite{HowardContextualityandComputation}, there
is some reason to believe that the mana has some privileged status
among all possible monotones. Determining if and how this intuition
can be formalized is a very important open problem.

There are a number of directions for future work, many of which have
already been discussed in the main body of the text. Other resource
theories admit a wealth of monotones. This is especially true in the
theory of entanglement where a large number of entanglement measures
have been developed to solve specialized problems. One obvious direction
for future work is the creation of additional magic monotones to address
particular problems in magic resource theory. It is also important
to develop the parts of the resource theory that are not encapsulated
by magic monotones. For example, analogues of entanglement catalysis
and activation are discussed in \cite{Campbell_magic_state_catalysis}.
The most urgent outstanding problem of this type is to find a criterion
for determining if it is possible to (asymptotically) reversibly convert
between particular resource states using stabilizer operations. Concretely,
it is always possible to use LOCC to reversibly convert pure bipartite
entangled states but this is not true for tripartite entanglement;
we would like to know which situation holds for magic theory. Even
a partial result of this type would be very powerful, offering deep
insight into the structure of stabilizer protocols. 

Much of this paper has been dedicated to showing that much of the technology from
other resource theories can be imported to the resource theory of
magic. It is very interesting to ask if we can go in the other direction
and export the insights of magic theory to the study of generic resource
theories and quantum theory broadly. One obvious extension of this
type is to the setting of linear optics, which is the infinite dimensional
analogue of the stabilizer formalism. Some progress on this front
has already been made: it has been shown that linear optics
operations acting on states with positive Wigner function, which includes non-Gaussian states, is efficiently 
classically simulable\cite{VeitchEtAl_LinOpticsWigSimulation,Mari2013_pos_wig_efficient_sim}.
We should also mention \cite{Kenfack2004Negativity} which 
examined the volume of the negative region of the infinite-dimensional
Wigner function as a measure of non-classicality but did not explore the resource
theory implications. 

The study of entanglement theory offers powerful insights into the
power of quantum communication protocols. This is because of the close
relationship between LOCC and quantum communication. Similarly, there
is a close relationship between stabilizers and quantum computation
beyond the application of stabilizer codes to fault-tolerant quantum
computation. The stabilizer operations are a maximal subset of efficiently
simulable quantum operations in the sense that the addition of any
pure non-stabilizer resource promotes stabilizer computation to universal
quantum computation\cite{CampbellMSDAllPrimeDim}. This suggests that
the usefulness of the tools developed here may extend beyond the study
of magic state computation to give insights into the origins of quantum
computational speedup. 

\begin{acknowledgements}
 We thank Marco Piani, Chris Ferrie, Robert Spekkens and Earl Campbell for helpful comments and discussions. The authors acknowledge financial
support from  CIFAR, USARO-DTO, and the Government of Canada through NSERC. Research at Perimeter Institute is supported by the Government of Canada through Industry Canada and by the Province of Ontario through the Ministry of Research and Innovation.
\end{acknowledgements}

\bibliographystyle{plain}
\bibliography{stab_resource,magic_state,infinited2}

\appendix
\section{Proofs on the relative entropy of magic}

We begin by showing that the relative entropy is a valid measure of
magic.

\subsubsection*{Relative entropy of magic is a monotone\label{sub:rel_ent_monotone}}
\begin{thm*}
The relative entropy of magic is a magic monotone.\end{thm*}
\begin{proof}
We need to verify that this function is non-increasing under stabilizer
operations.
\begin{enumerate}
\item Invariance under Clifford unitaries: For any unitary, $S\left(U\rho U^{\dagger}\|U\sigma U^{\dagger}\right)=S\left(\rho\|\sigma\right)$.
If $U$ is a Clifford and $\sigma$ is a stabilizer state then $U\sigma U^{\dagger}$
will also be a stabilizer state, ergo $\relent{U\rho U^{\dagger}}=\min_{\sigma}S\left(U\rho U^{\dagger}\|\sigma\right)=\min_{\sigma}S\left(U\rho U^{\dagger}\|U\sigma U^{\dagger}\right)=\min_{\sigma}S\left(\rho\|\sigma\right)=\relent{\rho}$.
\item Non-increasing on average under stabilizer measurement: Without loss of generality, we consider
computational basis measurement on the final qudit. Let $\{V_{i}\}=\{\mathbb{I}\otimes\ketbra ii\}$
be the measurement POVM and label outcome probabilities $p_{i}=\Tr\left(V_{i}\rho\right),\ q_{i}=\Tr\left(V_{i}\sigma\right)$
as well as post-measurement states $\rho_{i}=V_{i}\rho V_{i}^{\dagger}$
and $\sigma_{i}=V_{i}\sigma V_{i}^{\dagger}$. In reference \cite{VedralPlenioEntMeasandPurification}
it is shown that 
\[
\sum_{i}p_{i}S\left(\frac{\rho_{i}}{p_{i}} \Big\| \frac{\sigma_{i}}{q_{i}}\right)\le S\left(\rho\|\sigma\right).
\]
Since $\sigma_{i}/q_{i}$ is a stabilizer state whenever $\sigma$
is a stabilizer state this implies measurement does not increase the
relative entropy of magic on average.
\item Non-increasing under partial trace: From the strong subadditivity
property of the von Neumann entropy \cite{LiebStrongSubAddEntropy}
we have $S\left(\Tr_{B}\left(\rho\right)\|\Tr_{B}\left(\sigma\right)\right)\le S\left(\rho\|\sigma\right)$
from which the result follows immediately. 
\item Invariance under composition with stabilizer states: $S\left(\rho\otimes A\|\sigma\otimes A\right)=S\left(\rho\|\sigma\right)$
for any quantum state $A$, from which it follows $\relent{\rho\otimes A}\le\relent{\rho}$.
Equality follows because the relative entropy of magic is non-increasing
under the partial trace, i.e., $\relent{\rho}\le\relent{\rho\otimes A}$.
\end{enumerate}
\end{proof}
We now turn to the asymptotic variant of the relative entropy of magic,
$\regrelent{\rho}=\lim_{n\rightarrow\infty}\relent{\rho^{\otimes n}}/n$.
We show that this quantity is non-zero if and only if $\rho$ is a
magic state, which in particular implies that magic must be consumed
at a non-zero rate to create magic states. We will also need this
result for \thmref{asymp_conv_rate}.

\subsubsection*{Regularized relative entropy of magic is faithful.\label{sub:faithful-proof}}
\begin{thm*}
The regularized relative entropy of magic is faithful in the sense
that $\regrelent{\rho}=0$ if and only if $\rho$ may be written as
a convex combination of stabilizer states.\end{thm*}
\begin{proof}
We recover this result as a special case of the main theorem of reference
\cite{piani2009RelEntFaithful}. That paper introduces a variant of
the relative entropy measure that quantifies the distinguishability
of a quantum state from the set of free states using a restricted
set of measurements. Let $\{M_{i}\}$ be a measurement POVM and define
the map 
\[
\mathcal{M}\left(\rho\right)=\sum_{i}p_{i}\left(\rho\right)\ketbra ii,\ p_{i}\left(\rho\right)=\Tr\left(\rho M_{i}\right),
\]
where $\left\{ \ket i\right\} $ is any orthonormal set and $\mathcal{M}$
is a map associated to measurement $\{M_{i}\}$. Letting $\mathbb{M}$
be the set of restricted measurements we can define,
\[
\mathbb{M}S\left(\rho\|\sigma\right)\equiv\max_{\mathcal{M}\in\mathbb{M}}S\left(\mathcal{M}\left(\rho\right)\|\mathcal{M}\left(\sigma\right)\right).
\]
The significance of this quantity is from theorem 1 of \cite{piani2009RelEntFaithful}: 
\begin{thm*}
Consider a restricted set of operations inducing a resource theory.
Let $\mathbb{M}$ be the restricted set of measurements (here the
stabilizer measurements) and $P$ the set of free states (here the
stabilizer states). If the set of free states is closed under restricted
measurement and the partial trace then it holds that the regularization
of the relative entropy distance to the set of free states $r_{P}^{\infty}\left(\rho\right)$
satisfies
\[
r_{P}^{\infty}\left(\rho\right)\ge\min_{\sigma\in P}\mathbb{M}S\left(\rho\|\sigma\right).
\]

\end{thm*}

The stabilizer formalism satisfies the conditions of the
theorem. Moreover, since the stabilizer measurements contain an informationally
complete measurement it holds that $\mathbb{M}S\left(\rho\|\sigma\right)>0$
whenever $\rho$ is a magic state. This implies $\regrelent{\rho}>0$
whenever $\rho$ is a magic state. $\regrelent{\rho}=0$ for all stabilizer
states $\rho$, so the claimed result follows.

\end{proof}

\section{Proofs on sum negativity and mana}

\subsection{Odd dimensions}

The main ingredient in establishing both $\sn{\rho}$ and $\mana{\rho}$
as magic monotones is to show that $\wignorm{\rho}=\sum_{\boldsymbol{u}}\abs{W_{\rho}\left(\boldsymbol{u}\right)}$
is a magic monotone.

\subsubsection*{Wigner function 1-norm is a magic monotone.\label{sub:wig_norm_monotone_proof}}
\begin{thm*}
$\wignorm{\rho}=\sum_{\boldsymbol{u}}\abs{W_{\rho}\left(\boldsymbol{u}\right)}$
is a convex magic monotone. \end{thm*}
\begin{proof}
We need to verify that this function is non-increasing under stabilizer
operations:
\begin{enumerate}
\item Invariance under Clifford unitaries: The action of Clifford unitaries
on the phase space of the Wigner function is a permutation, $\boldsymbol{u}\rightarrow F\boldsymbol{u}$.
Thus, $\wignorm{U\rho U^{\dagger}}=\sum_{\boldsymbol{u}}\abs{W_{U\rho U^{\dagger}}\left(\boldsymbol{u}\right)}=\sum_{\boldsymbol{u}}\abs{W_{\rho}\left(F\boldsymbol{u}\right)}=\sum_{\boldsymbol{u}}\abs{W_{\rho}\left(\boldsymbol{u}\right)}=\wignorm{\rho}$.
\item Non-increasing on average under stabilizer measurement: We consider
computational basis measurement on the final qudit. The expected value
of $\wignorm{\tilde{\rho}}$ for the post measurement state $\tilde{\rho}$
is: 
\begin{eqnarray*}
\mathbb{E}\left[\wignorm{\tilde{\rho}}\right] & = & \sum_{i}\Tr\left(\rho\mathbb{I}\otimes\ketbra ii\right)\wignorm{\left(\mathbb{I}\otimes\ketbra ii\right) \rho \left(\mathbb{I}\otimes\ketbra ii\right)/\Tr\left(\rho\mathbb{I}\otimes\ketbra ii\right)}\\
 & = & \sum_{i}\wignorm{\left(\mathbb{I}\otimes\ketbra ii\right) \rho\left(\mathbb{I}\otimes\ketbra ii\right)},
\end{eqnarray*}
and by writing $\left(\mathbb{I}\otimes\ketbra ii\right)\rho\left(\mathbb{I}\otimes\ketbra ii\right)$
as:
\begin{eqnarray*}
\left(\mathbb{I}\otimes\ketbra ii\right)\rho \left(\mathbb{I}\otimes\ketbra ii\right) & = & \sum_{\boldsymbol{u},\boldsymbol{v}}W_{\rho}\left(\boldsymbol{u}\oplus\boldsymbol{v}\right)\braopket i{A_{\boldsymbol{v}}}i\cdot A_{\boldsymbol{u}}\otimes\ketbra ii\\
 & = & \sum_{\boldsymbol{u}}\left(\sum_{\boldsymbol{v}}W_{\rho}\left(\boldsymbol{u}\oplus\boldsymbol{v}\right)\braopket i{A_{\boldsymbol{v}}}i\right)A_{\boldsymbol{u}}\otimes\sum_{\boldsymbol{w}}\left(\frac{1}{d}\braopket i{A_{\boldsymbol{w}}}i\right)A_{\boldsymbol{w}}
\end{eqnarray*}
we find,
\begin{eqnarray*}
\mathbb{E}\left[\wignorm{\tilde{\rho}}\right] & = & \sum_{i}\sum_{\boldsymbol{u},\boldsymbol{w}}\abs{\left(\sum_{\boldsymbol{v}}W_{\rho}\left(\boldsymbol{u}\oplus\boldsymbol{v}\right)\braopket i{A_{\boldsymbol{v}}}i\right)\left(\frac{1}{d}\braopket i{A_{\boldsymbol{w}}}i\right)}\\
 & = & \sum_{i}\sum_{\boldsymbol{u}}\left(\sum_{\boldsymbol{w}}\frac{1}{d}\braopket i{A_{\boldsymbol{w}}}i\right)\abs{\left(\sum_{\boldsymbol{v}}W_{\rho}\left(\boldsymbol{u}\oplus\boldsymbol{v}\right)\braopket i{A_{\boldsymbol{v}}}i\right)}\ \ \text{(}\because\braopket i{A_{\boldsymbol{w}}}i\ge0\text{)}\\
 & \le & \sum_{i}\sum_{\boldsymbol{u}}\sum_{\boldsymbol{v}}\abs{W_{\rho}\left(\boldsymbol{u}\oplus\boldsymbol{v}\right)\braopket i{A_{\boldsymbol{v}}}i}\text{ \ \ (\ensuremath{\because} triangle inequality and $\sum_{\boldsymbol{w}}\frac{1}{d}\braopket i{A_{\boldsymbol{w}}}i=1$)}\\
 & = & \sum_{\boldsymbol{u},\boldsymbol{v}}\left(\sum_{i}\braopket i{A_{\boldsymbol{v}}}i\right)\abs{W_{\rho}\left(\boldsymbol{u}\oplus\boldsymbol{v}\right)}\text{ \ \ (}\because\braopket i{A_{\boldsymbol{w}}}i\ge0\text{)}\\
 & = & \wignorm{\rho}\text{ \ \ (}\because\sum_{i}\braopket i{A_{\boldsymbol{v}}}i=1).
\end{eqnarray*}

\item Invariance under composition with stabilizer states: Let $\sigma$
be any state with positive Wigner representation. Then, 
\begin{eqnarray*}
\wignorm{\rho\otimes\sigma} & = & \wignorm{\rho}\wignorm{\sigma}\\
 & = & \wignorm{\rho},
\end{eqnarray*}
since $\wignorm{\sigma}=\sum_{\boldsymbol{u}}\abs{W_{\sigma}\left(\boldsymbol{u}\right)}=\sum_{\boldsymbol{u}}W_{\sigma}\left(\boldsymbol{u}\right)=1$
for positively represented states. All stabilizer states are positively
represented so they are included as a special case.
\item Non-increasing under partial trace: We trace out the final
qudit $B$ of the system. If $\rho=\sum_{\boldsymbol{u},\boldsymbol{v}}W_{\rho}\left(\boldsymbol{u}\oplus\boldsymbol{v}\right)A_{\boldsymbol{u}}\otimes A_{\boldsymbol{v}}$
then $\Tr_{B}\left(\rho\right)=\sum_{\boldsymbol{u}}\left(\sum_{\boldsymbol{v}}W_{\rho}\left(\boldsymbol{u}\oplus\boldsymbol{v}\right)\right)A_{\boldsymbol{u}}$,
so
\begin{eqnarray*}
\wignorm{\Tr_{B}\left(\rho\right)} & = & \sum_{\boldsymbol{u}}\abs{\sum_{\boldsymbol{v}}W_{\rho}\left(\boldsymbol{u}\oplus\boldsymbol{v}\right)}\\
 & \le & \wignorm{\rho},
\end{eqnarray*}
by the triangle inequality.
\item Convexity: 
\begin{eqnarray*}
\wignorm{p\rho+\left(1-p\right)\sigma} & = & \sum_{\boldsymbol{u}}\abs{pW_{\rho}\left(\boldsymbol{u}\right)+\left(1-p\right)W_{\sigma}\left(\boldsymbol{u}\right)}\\
 & \le & p\wignorm{\rho}+\left(1-p\right)\wignorm{\sigma},
\end{eqnarray*}
 by the triangle inequality.
\end{enumerate}
\end{proof}
We next establish that this was essentially the only choice we could
have made to (simply) quantify the magic of a quantum state via its
Wigner representation.

\subsubsection*{Sum negativity is the unique phase space measure of magic.\label{sub:sum_neg_unique_proof}}
\begin{thm*}
Assume $\mono{\rho}$ is a function on quantum states that satisfies
the following conditions: 1. $\mono{\rho}$ is a magic monotone, 2.
$\mono{\rho}$ is determined only by the negative values of the Wigner
function and 3. $\mono{\rho}$ is invariant under arbitrary permutations
of discrete phase space (that is, even under permutations that do
not correspond to quantum transformations). Then $\mono{\rho}$ may
be written as a function of only $\sn{\rho}$. \end{thm*}
\begin{proof}
Let $\rho$ have negative entries $-N_{1},-N_{2},\dots,-N_{k}$ and
$\rho'$ have negative entries $-N_{1}',-N_{2}',\dots,-N_{k'}'$,
with

\[
N\equiv\sn{\rho}=\sum N_{i}=\sum N_{i}'=\sn{\rho'}.
\]
$A$ and $A'$ will be ancilla states acting on $m$ qudits, with
$m=\max\left\{ \lceil\log_{d}k\rceil,\lceil\log_{d}k'\rceil\right\} $;
$d$ is the size of each qudit.

\begin{eqnarray*}
A & = & \sum_{i=1}^{k'}(N_{i}'/N)\ketbra ii\\
A' & = & \sum_{i=1}^{k}(N_{i}/N)\ketbra ii.
\end{eqnarray*}
These are valid states since the sum of the $N_{i}$ and $N_{i}'$
is the same. The Wigner function of $A$ consists of columns labeled
by $i$ with entries $N_{i}'/rN$, with $r=d^{m}$; each column contains
$r$ such elements. It also has $d^{m}-k'$ columns filled with zeros.
Similarly for $A'$, except it has $d^{m}-k$ zero columns and the
non-zero columns have $r$ copies of $N_{i}/rN$ instead.

The negative Wigner function entries for the state $\rho\otimes A$
are of the form $-N_{i}N_{j}'/rN$, for all $i$ and $j$. Each of
these appears $r$ times. The negative Wigner function entries for
$\rho'\otimes A'$ are of the form $-N_{j}'N_{i}/rN$, for all i and
j. Again, each appears $r$ times. These entries could be in different
locations, but since the function we are calculating does not depend
on location of negative entries, only their values, it follows that
\[
\mono{\rho}=\mono{\rho\otimes A}=\mono{\rho'\otimes A'}=\mono{\rho'}.
\]
Therefore, $\mono{\rho}$ is a function only of $\sn{\rho}$.
\end{proof}

\subsection{Continuity and Asymptotic Continuity\label{sub:mojo_continuity}}

In practice a perfect conversion is generally not possible, $\|\Lambda\left(\rho^{\otimes m}\right)-\sigma^{\otimes n}\|_{1}>0$
for even the best choice of stabilizer protocol $\Lambda$. A state
$\tilde{\sigma}_{n}$ that is close enough to $\sigma^{\otimes n}$
can be used in place of $\sigma^{\otimes n}$ in information theoretic
tasks so a better notion of conversion would be: how many copies of
$\rho$ are required to produce a state $\Lambda(\rho^{\otimes m})=\tilde{\sigma}_{n}$
that is ``close enough'' to $\sigma^{\otimes n}$. A natural notion
of closeness is $\|\tilde{\sigma}_{n}-\sigma^{\otimes n}\|_{1}<\epsilon$
for some operationally relevant $\epsilon$. It is conceivable that
there is some choice of $\tilde{\sigma}_{n}$ in the epsilon ball
around $\sigma^{\otimes n}$ such that $\mana{\tilde{\sigma}_{n}}\ll\mana{\sigma^{\otimes n}}$,
in which case $\mana{\sigma}$ would have little operational significance.
Happily, it is not difficult to show that $\mana{\rho}$ is continuous
with respect to the 1-norm in the sense that for a sequence of states
$\rho_{k},\sigma_{k}\in\mathcal{S}\left(\mathcal{H}_{d}\right)$ $\left\{ \|\rho_{k}-\sigma_{k}\|\right\} _{k}\rightarrow0\implies\left\{ \abs{\mana{\rho_{k}}-\mana{\sigma_{k}}}\right\} _{k}\rightarrow0$,
so for a target state \emph{of fixed dimension }there is some well-defined sense in which closeness in the 1-norm implies that the mana
of two states is close.

In the case of asymptotic conversion of states this notion needs some
massaging. Formally, let $\Lambda_{n}:\mathcal{S}\left(\mathcal{H}_{d^{m(n)}}\right)\rightarrow\mathcal{S}\left(\mathcal{H}_{d^{n}}\right)$
be stabilizer protocols satisfying 
\begin{eqnarray*}
\lim_{n\rightarrow\infty}\|\Lambda_{n}\left(\rho^{\otimes m(n)}\right)-\sigma^{\otimes n}\| & \rightarrow & 0
\end{eqnarray*}
In particular we would like to avoid a situation where $\lim_{n\rightarrow\infty}\mana{\Lambda_{n}\left(\rho^{\otimes m(n)}\right)}\ll\mana{\sigma^{\otimes n}}$.
One way that this requirement can be formalized is the property of
asymptotic continuity. A function is said to be asymptotically continuous
if for sequences $\rho_{n},\sigma_{n}$ on $\mathcal{H}_{n}$, $\lim_{n\rightarrow\infty}\|\rho_{n}-\sigma_{n}\|\rightarrow0$
implies: 
\[
\lim_{n\rightarrow\infty}\frac{f(\rho_{n})-f(\sigma_{n})}{1+\log(\dim\mathcal{H}_{n})}\rightarrow0.
\]
This notion is the commonly accepted generalization of continuity
to the asymptotic regime and is of particular importance because if
the mana could be shown to be asymptotically continuous it would give the asymptotic conversion rate, as in \thmref{asymp_conv_rate}.
Unhappily, it is very difficult to show this. This is mostly because
it is false:
\begin{thm*}
$\mana{\sigma}$ is not asymptotically continuous.\end{thm*}
\begin{proof}
Define $\tilde{\sigma}_{n}=\left(1-\delta_{n}\right)\sigma^{\otimes n}+\delta_{n}\eta_{n}$, with $\lim_{n\rightarrow\infty}\delta_{n}\rightarrow0$.
Asymptotic continuity would imply 
\[
\lim_{n\rightarrow\infty}\frac{\mana{\tilde{\sigma}_{n}}-\mana{\sigma^{\otimes n}}}{n}\rightarrow0, 
\]
but
we will show this need not be the case. Suppose $\sigma$ is negative
on points $\mathcal{N}=\left\{ \boldsymbol{u}:\ W_{\sigma}(\boldsymbol{u})<0\right\} $.
Let $\eta$ be the state with maximal sum negativity satisfying $W_{\eta}(\boldsymbol{u})<0\iff\boldsymbol{u}\in\mathcal{N}$
(ie. $\eta$ is negative on the same points as $\sigma$). Then,
\begin{eqnarray*}
\wignorm{\tilde{\sigma}} & = & \sum_{\boldsymbol{u}}|\left(1-\delta_{n}\right)W_{\sigma^{\otimes n}}(\boldsymbol{u})+\delta_{n}W_{\eta^{\otimes n}}(\boldsymbol{u})|\\
 & = & \sum_{\boldsymbol{u}}\left(\left(1-\delta_{n}\right)|W_{\sigma^{\otimes n}}(\boldsymbol{u})|+\delta_{n}|W_{\eta^{\otimes n}}(\boldsymbol{u})|\right)\\
 & = & \left(1-\delta_{n}\right)\wignorm{\sigma^{\otimes n}}+\delta_{n}\wignorm{\eta^{\otimes n}}\\
 & = & \left(1-\delta_{n}\right)\wignorm{\sigma}^{n}+\delta_{n}\wignorm{\eta}^{n}.
\end{eqnarray*}
Here we have exploited that the sign of $W_{\eta^{\otimes n}}(\boldsymbol{u})$
and the sign of $W_{\sigma^{\otimes n}}(\boldsymbol{u})$ are always
the same. Subbing this in,
\[
\frac{\mana{\tilde{\sigma}_{n}}-\mana{\sigma^{\otimes n}}}{n}=\frac{1}{n}\log\left(\left(1-\delta_{n}\right)+\delta_{n}\left(\frac{\wignorm{\eta}}{\wignorm{\sigma}}\right)^{n}\right),
\]
but by assumption $\wignorm{\eta}>\wignorm{\sigma}$ unless $\wignorm{\sigma}$
is maximal for all states that are negative on $\mathcal{N}$, so
the limit need not go to $0$. Thus asymptotic continuity can not
hold generally.
\end{proof}
This result is not actually terribly surprising. Suppose we have a
preparation apparatus that always prepares $\sigma^{\otimes n}$.
Now further suppose that we rebuild our apparatus so that with probability
$\delta_{n}$ it will instead produce $\eta^{\otimes n}$ \emph{with
a far greater amount of negativity}. Then it is intuitively obvious
that we should be able to extract more negativity from the new apparatus
just by sacrificing a few copies of the output state to determine
whether we have produced $\sigma$ or $\eta$. Of course as $n$ goes
to infinity this will only work if $\delta_{n}$ goes to zero slowly
enough, but this argument does clarify the irrelevance of asymptotic
continuity. 

Essentially asymptotic continuity fails because it is possible that
access to a very large amount of resource, even with small probability,
can dramatically improve our preparation procedure. Notice that the
opposite is not (obviously) true: if our machine fails with a very
small probability this does not make it useless. Indeed, if we had
a promise of the form $\tilde{\sigma}_{n}=\left(1-\delta\right)\sigma^{\otimes n}+\delta\eta^{\otimes n}$
then we could just sacrifice some small number of registers to check
that that the output state was in fact $\sigma^{\otimes n}$. 
\end{document}